\documentclass[10pt,a4paper]{article}

\usepackage{amsthm}
\usepackage{amsmath}
\usepackage{amssymb} %added
\usepackage{float} %added
\usepackage{multirow} %added
\usepackage{tabularx} %added
\newcolumntype{K}[1]{>{\centering\arraybackslash}p{#1}} %added
\usepackage{natbib}
\usepackage[colorlinks,citecolor=blue,urlcolor=blue,filecolor=blue,backref=page]{hyperref}
\usepackage{graphicx}

\usepackage[title]{appendix}

\usepackage{tikz} %added
\usetikzlibrary{arrows} %added
\usetikzlibrary{shapes} %added
\usetikzlibrary{automata} %added
\usetikzlibrary{positioning} %added
\usetikzlibrary{backgrounds} %added
\usetikzlibrary{fit} %added

\newtheorem{theorem}{Theorem}[section]
\newtheorem{lemma}[theorem]{Lemma}
\newtheorem{corollary}{Corollary}[theorem]
\theoremstyle{definition}
\newtheorem{definition}{Definition}[section]

\newtheorem{example}{Example}[section]

\begin{document}
\title{Bayesian Networks, Total Variation and Robustness}
\author{Sophia K. Wright \& Jim Q. Smith}

\date{\small{Department of Statistics, University of Warwick, Coventry, CV4 7AL, UK}}

\maketitle

\begin{abstract}
\noindent Now that Bayesian Networks (BNs) have become widely used, an appreciation is developing of just how critical an awareness of the sensitivity and robustness of certain target variables are to changes in the model.  When time resources are limited, such issues impact directly on the chosen level of complexity of the BN as well as the quantity of missing probabilities we are able to elicit.  Currently most such analyses are performed once the whole BN has been elicited and are based on Kullback-Leibler information measures.  In this paper we argue that robustness methods based instead on the familiar total variation distance provide simple and more useful bounds on robustness to misspecification which are both formally justifiable and transparent.  We demonstrate how such formal robustness considerations can be embedded within the process of building a BN.  Here we focus on two particular choices a modeller needs to make: the choice of the parents of each node and the number of levels to choose for each variable within the system.  Our analyses are illustrated throughout using two BNs drawn from the recent literature.  
\end{abstract}

\noindent The first author would like to acknowledge their EPSRC grant EP/L016710/1, awarded through the OxWaSP CDT programme. The second author contributed to this work within their EPSRC grant EP/N510129/1.

\section{Introduction}
Bayesian Networks (BNs) are now a widely used probabilistic modelling tool, particularly in the field of decision support.  It is now acknowledged as best practice \citep{cowell_dawid_lauritzen_spiegelhalter, laskey_mahoney_2000, jim_book} that these models are set up in two distinct stages.  Firstly the structure of the BN, as expressed by its Directed Acyclic Graph (DAG), is either directly elicited from domain experts or when sufficient supporting data exists, learned from the data using a model search algorithm with default priors on the hyperparameters, see \cite{boneh_2010} and \cite{korb-nicholson2010}.  Once this graphical framework has been discovered, the graph is embellished into a full probabilistic model.  In the case of a discrete BN, this second stage involves eliciting or estimating, using priors on probabilities informed by expert judgements, the entries of its conditional probability tables (CPTs).  These CPTs provide the numerical prespecification of all the conditional probabilities needed to generate the full joint probability mass function and hence a fully specified probability model.

When engaging in this two stage process the analyst needs to be fully aware of precisely which inputs might be critical to the inferences made through the BN, see \cite{Albrecht2014}.  One critical element in an elicitation, or statistical estimation of the graph is to ensure these critical features are specified as accurately as possible.  This is especially important when elicitation or estimation is resource limited, as is usually the case in practice.  The modeller can then optimise their allocation of resources to concentrate on eliciting those elements of the model whose misspecification might most influence the required outputs.

To this end, the practitioner, prompted by the functionality of various software, is currently encouraged to develop awareness of the robustness of a chosen model to its inputs by performing a one-at-a-time numerical sensitivity analysis of the preliminary BN.  Here various different forms of numerical contaminations of the model are investigated, where effects are usually measured in terms of mutual information/Kullback-Leibler divergence \citep{albrecht_nicholson_whittle, friedman, Nicholson_mutual_information, Zaragoza}.  This type of study is obviously extremely useful.  On the other hand it has drawbacks.  First, it relies on the chosen enacted perturbations covering the entire space which becomes more challenging as models become increasingly large.  Furthermore, even if such a search is performed systematically, the impacts (most currently measured by mutual information), are not directly relevant to the impact on ensuing decisions, see below for further clarification.  Additionally, such an analysis must perforce be performed after the model has been fully specified.  This means that the whole probability model is needed before the sensitivity analysis can be performed.  One interesting recent attempt to provide such assessments after the structural elicitation phase, but before the probabilistic embellishment is through the use of distance weighted sensitivity measures \citep[see][]{albrecht_nicholson_whittle}.  However, these do not dovetail with the mutual information measures described above and have a level of arbitrariness in the choice of weight function needed to use this method. 

Over recent years more formal and systematic \emph{robustness} analyses have appeared.  Robustness of probability models has been studied by statisticians for many decades, and specific methodology for Bayesian Networks has also been recently developed:  \cite{coupe_gaag_2002}, \cite{BN_KL}, \cite{laskey-1995}, \cite{o'neill-2009}, \cite{renooij-2010}.  These fall into two main streams: local robustness studies and global studies.  In the former, a chosen probability model is perturbed using a finite parametrised modification.  The latter, termed global analyses, does not rely on perturbation lying within a given parametric family \citep{o'neill-2009, smith-daneshkhah2010}.  Instead, an appropriate divergence measure is applied to first specify a neighbourhood system around each model.  Bounds are then calculated for the maximum deviation in the inference that could be achieved by a model in this neighbourhood.  If this deviation is small then the model is deemed to be robust \citep{Gustafson&Wasserman,smith-rigat2012}.  Both types of robustness analysis have been applied to BNs in work such as \cite{smith-daneshkhah2010}.  In this paper we focus solely on global robustness studies as applied to finite discrete BNs.

Thus far, global robustness studies for BNs have mainly centred around the analysis of how robust a model might be to perturbations, with respect to Kullback-Leibler  (KL) or Chan-Darwiche divergences \citep[see][]{chan_darwiche_2005, BN_KL, leonelli_gorgen_smith}.  Both of these divergence measures benefit from some helpful technical properties which allow various measures of dependence to have explicit formulae.  These measures are specified in terms of log probabilities in the KL case or equivalently ratios of probabilities in the Chan-Darwiche instance.  Therefore, both have the disadvantage that they depend very heavily on the accurate specification of \emph{small probabilities}.  However, it is well documented that it is precisely these small probabilities that typically exhibit the largest elicitation error \citep[see][]{elicitation_book, jim_book}.  Furthermore, when BNs are learned from data, any associated small probabilities are difficult to reliably estimate from data, because almost by definition we will see very few of these events in any training set we use to estimate a model. 

In many circumstances (especially in decision analysis), the misspecification of improbable event probabilities has only a small impact on the required outputs of a decision analysis: see below.  For the purposes of the two stage process described above, the Kullback-Leibler and Chan-Darwiche divergence measures are hardly ideal as they can be highly sensitive to misspecification which may have little effect on any supported decision analysis. 

In this paper we demonstrate that an alternative robustness study based on a more conventional divergence measure (widely used in probability theory and stochastic analysis), which is the total variation distance, has some serious practical and theoretical advantages over its main competitors.  Although it is often difficult to derive \emph{explicit} formulae for the impacts of deviation in variation, it is nevertheless straightforward to \emph{tightly bound} such deviations in variation distance.  Deviation in variation corresponds much more closely to the types of error we would envisage experiencing within either an elicitation exercise or through misestimation.  Perhaps most important, the expectation of a fixed bounded utility function $U$, under various decisions (induced by an approximation) are simply bounded by linear functions of the total variation in the probability distributions of the attributes of $U$ \citep[see e.g.][]{jim_book}.  Note that in a BN these attributes will typically constitute a small subset of the totality of its variables.  Hence small variation distances (between probability mass functions)  on these small subsets translate \textit{directly} into small effects in the pertinent expected utilities.  Conversely, large deviations translate into large effects that might have a greater impact on some specification of a utility.  

In the following section we review the BN framework and introduce our examples.  Then in Section \ref{TVdist_section} we review some simple properties of the total variation distance and show that the effect in variation distance of the misspecification of the probability mass function of one random variable in a BN to another diminishes exponentially.  We then discover explicit bounds for this error both when the BN is decomposable and more generally.  We demonstrate that this impact can be bounded explicitly in terms of a simple function of the extreme entries of the CPTs within the BN.  These results have the useful spin-off that CPTs do not necessarily need to be fully elicited before the robustness analysis can take place.  In Section \ref{approximations_section} we show how these explicit measures of robustness can be applied to determine the effect of approximating simplifications on the topology of the BN and additionally, to decide the number of levels into which to categorise each variable.  We demonstrate how by using total variation, robustness analyses can be performed in a harmonious composite way that directly bounds the impact on decision making of various types of expedient approximations.  Finally in Section \ref{strategies_section} we provide some guidelines to best employ our results in practice and discuss some enhancements of our strategy.

\section{Hypotheses of a Bayesian Network}
\label{BN section}
We begin by giving a short review of BNs and some of its properties we use later in the paper.  A discrete Bayesian Network (BN) $\mathcal{G}$, or DAG,  on a random vector $\boldsymbol{X}\triangleq \left( X_1,X_2,\ldots,X_m\right) $ represents a family of models which respect a set of conditional independence hypotheses so that for $i=2,3,\ldots ,m$ 
\begin{align*}
X_i \amalg \boldsymbol{X}_{R(i)}|\boldsymbol{X}_{P_a(i)} 
\end{align*}
where $Pa(i)\subseteq \left\{ 1,2,\ldots ,i-1\right\} $ are the \emph{parents} of $X_i$, i.e. those indices of the previously listed variables on which $X_i$ depends, and $R(i)\triangleq \left\{ 1,2,\ldots ,i-1\right\} \backslash Pa(i)$. 

An equivalent expression is that the joint probability mass function $p(\boldsymbol{x})$ of $\boldsymbol{X}$ factorises as 
\begin{equation}
p(\boldsymbol{x})=p(x_1)p(x_2|\boldsymbol{x}_{P_a(2)})\ldots p(x_{i}|
\boldsymbol{x}_{P_{a}(i)})\ldots p(x_{m}|\boldsymbol{x}_{P_a(m)}).
\label{BN fact}
\end{equation}
In either formulation the directed graph of the BN has vertex set $\left\{ X_1,X_2,\ldots ,X_m \right\} $ and a directed edge from $X_j$ to $X_i$ iff $j\in Pa(i)$. 

An important subclass of BNs whose properties we discuss later, are those which are called decomposable.  A \emph{decomposable} BN $\mathcal{G}$ is a BN in which every parent set of each node $X_i$ in the graph forms a complete subgraph of $\mathcal{G}$.  It is simple to show that any BN can (albeit inefficiently) be re-expressed in a decomposable BN containing it \citep{lauritzen1996, korb-nicholson2010, jim_book}.  This property, widely used for propagation algorithms, can also be used for robustness analyses. 

When a BN is decomposable it can be shown \citep[see][]{lauritzen1996, jim_book} that the joint density factors in the following way.  The \emph{cliques} $\left\{ \boldsymbol{X}_{C_1},\boldsymbol{X}_{C_2},\boldsymbol{X}_{C_3},\ldots ,\boldsymbol{X}_{C_m}\right\} $ i.e. the maximal connected subsets of the decomposable graph $\mathcal{G}$ can be totally ordered starting with any clique, label this $\boldsymbol{X}_{C_1}$.  We call $\boldsymbol{X}_{S(i)}$ where $S_i=C_i  \cap \cup_{k=1}^{i-1} C_k$ the \emph{separator} of $\boldsymbol{X}_{C_i}$ from $\cup_{k=1}^{i-1} C_k$. An indexing is said to satisfy the \emph{running intersection property} if for all $i=1,2,\hdots,m$ there exists some index $j < i$ such that $S_i = C_i \cap \cup_{k=1}^{i-1} C_k \subseteq C_j$. This implies that the result of intersecting a clique with all previous cliques is contained within one or more earlier cliques \cite{lauritzen1996, jim_book}.  Note the choice of $C_j$ may not be unique

We can depict one of these choices of order and containment by a \emph{junction tree} $\mathcal{J(G)}$.  This is an undirected tree whose vertices are $\left\{ \boldsymbol{X}_{C_1},\boldsymbol{X}_{C_2},\boldsymbol{X}_{C_3},\ldots \boldsymbol{X}_{C_m}\right\} $ and whose $m-1$ undirected edges simply connect $\boldsymbol{X}_{C_i}$\ to $\boldsymbol{X}_{C_{j(i)}}, i=2,3,\ldots ,m$.  Note that these edges can be labelled by a corresponding separator of $\mathcal{G}$.  Here we will for simplicity assume that the entries of the joint mass function $p(\boldsymbol{x})$ are all strictly positive, although this is not strictly necessary \citep[see][]{lauritzen_spiegelhalter_1988}.  In fact this is advised from a practical point of view by a number of authors e.g. \cite{korb-nicholson2010} when dealing with no known functional relations.  It can then be proved \citep[e.g. see][]{cowell_et_al_2007, jim_book} that $p(\boldsymbol{x})$ of any such decomposable BN $\mathcal{G}$ respects the following formula:
\begin{equation*}
p(\boldsymbol{x})=\frac{p(\boldsymbol{x}_{C_1}).p(\boldsymbol{x}_{C_2}).p(\boldsymbol{x}_{C_3})\ldots p(\boldsymbol{x}_{C_m})}{p(\boldsymbol{x}_{S_2}).p(\boldsymbol{x}_{S_3})\ldots p(\boldsymbol{x}_{S_m})}.  \label{decompratfn}
\end{equation*}
One straightforward but important consequence of this decomposition used later is that given any BN $\mathcal{G}$ and an associated junction tree $\mathcal{J(\mathcal{G})}$, then  for any two cliques $\boldsymbol{X}_{C_1},\boldsymbol{X}_{C_k}$ there is a \emph{unique} sequence of cliques $\left( \boldsymbol{X}_{C_1},\boldsymbol{X}_{C_2},\boldsymbol{X}_{C_3},\ldots ,\boldsymbol{X}_{C_k}\right) $ with no repeats, and separators $\left( \boldsymbol{X}_{S_2},\boldsymbol{X}_{S_3},\ldots ,\boldsymbol{X}_{S_k}\right) $ between $\boldsymbol{X}_{C_1}$ and $\boldsymbol{X}_{C_k}$ within $\mathcal{J(\mathcal{G})}$, called a \textit{simple path}.  If we write $\overline{C}_k\triangleq \cup _{j=1}^{k}C_j$, then since $S_i \subseteq C_i$ we know that $\boldsymbol{x}_{S_i}$ is a subvector of $\boldsymbol{x}_{C_i}$ giving 
\begin{align}
p(\boldsymbol{x}_{\overline{C}_k}) &=\frac{p(\boldsymbol{x}_{C_1}).p(\boldsymbol{x}_{C_2}).p(\boldsymbol{x}_{C_3}),\ldots p(\boldsymbol{x}_{C_k})}{p(\boldsymbol{x}_{S_2}).p(\boldsymbol{x}_{S_3}),\ldots p(\boldsymbol{x}_{S_k})}  \label{decomposable sequence} \\
&=p(\boldsymbol{x}_{C_1}).p(\boldsymbol{x}_{C_2}|\boldsymbol{x}_{S_2}).p(\boldsymbol{x}_{C_3}|\boldsymbol{x}_{S_3})\ldots p(\boldsymbol{x}_{C_k}|\boldsymbol{x}_{S_k}).  \nonumber 
\end{align}
\begin{lemma}
\label{junction tree altered}
It follows from Equation \ref{decomposable sequence} and the conditional independence in $\mathcal{G}$ that if $T_{1,k}\triangleq \cup _{i=3}^{k-1}S_i$ then
\begin{equation*}
p(\boldsymbol{x}_{C_1\cup C_k})=\sum_{\boldsymbol{x}_{T_{1,k}}}p(\boldsymbol{x}_{C_1}).p(\boldsymbol{x}_{S_3}|\boldsymbol{x}_{S_2}).p(\boldsymbol{x}_{S_4}|\boldsymbol{x}_{S_3})\ldots p(\boldsymbol{x}_{S_k}|\boldsymbol{x}_{S_{k-1}})p(\boldsymbol{x}_{C_k}|\boldsymbol{x}_{S_k}).  \label{JTresult}
\end{equation*}
\end{lemma}
Thus we have a formula for the joint mass function of a ``donating'' clique $C_1$ and a ``target'' clique $C_k$ depending on $C_1$, expressed in terms of a sequence of transitions in a non-homogeneous Markov Chain.  Although this property derives directly from the elementary properties of trees it is important, and an often overlooked property.  It means that standard results from non-homogeneous Markov Chain theory can be used to measure the extent of the diminishing effect of information as it passes along this simple path.  In particular it is well-known that variation distance in an ergodic, acyclic Markov Chain contracts as information is propagated through the system.  The observation in Lemma \ref{junction tree altered} is therefore critical to the development of some of the robustness bounds we develop here.

\subsection{Applying a Bayesian Network in Practice}

A BN is generally selected in one of two ways.  Occasionally we may have access to a complete training data set from which we can select the most promising explanatory BN whose associated $p(\boldsymbol{x})$ respecting Equation \ref{BN fact} appears to best fit the data.  There are many ways to do this, including using software packages such as `bnlearn' in R \citep[see][]{scutari_book}.  However, when applying such a model selection method in practice, even for low dimensional BNs, it is common to find many models score similarly well.  A BN may not adequately describe all features in the data set.  Even if we know this model to be true, as in a simulation exercise or even a moderately sized problem, it has been demonstrated that the best model is only close to the generating process, unless the training data set is absolutely enormous \cite{cussens2011}.  There are also the obvious statistical errors associated with the representativeness of the data set used, even if sampling is performed at random.  Hence it is rare for a single data generating model to be unequivocally identified. Considering the robustness of the critical outputs of the fitted model is therefore a critical element of any ensuing statistical analysis. 

The second way to create a BN is by performing a direct elicitation from an expert.  Here, having listed the variables in an order which might be compatible with the sequence in which those measurements may occur, the expert is asked for each $X_i$ ($i=2,3,\ldots ,n$) of the previously mentioned variables which might be relevant to forecasting it.  Building on this qualitative framework, hopefully faithful to the expert's actual judgements, we then proceed to embellish the graph by supplementing the structure with the specification of the corresponding CPTs $\left\{ p(x_i|\boldsymbol{x}_{P_a(i)}):i=1,2,3,\ldots ,m\right\}$.  These probabilities will be subject to elicitation error, although the preceding structural elicitation process aims to mitigate this specification error \citep{korb-nicholson2010, EFSA, jim_book}.  Again an understanding of the robustness to perturbations of the hypothesised graphical framework and also the entries in the CPTs of any inferential assumptions we make here, will clearly be critical to a good statistical analysis. 

\subsection{Applications}
\subsubsection{Food Security System}
To illustrate the uses and practicalities of our results we shall be using the Food Security Integrated Decision Support System (IDSS) described in greater depth in \cite{bookchapter} and \cite{SBL2015}.  The aim of this massive multi-layered dynamic BN is to ascertain which local government policies influence or improve the level of food poverty within the UK.  However, the targeted user here is primarily interested in three specific classes of outputs: Health (of its local constituents), Educational Attainment of children and measures of Social Cohesion \citep{SBL2015}, measures of which, in the terminology of this paper, will form our final vector of target variables.  

The overarching IDSS model is a DBN model as shown in Figure \ref{IDSS_levels} in which the target nodes are classified as Level 1.  Each component of this model can be broken down into detailed subnetworks.  For example, the Level 2 `UK Food Costs' depends on the availability of food, production costs and so on.  Specifically we may be interested in access to healthy food necessities such as fruit and vegetables which rely heavily on pollinator abundance.  A sub-subnetwork to determine the factors which affect the pollinator abundance is therefore required, a fragment of which is shown in Figure \ref{Pollinator_BN}.  A subset of this BN has been elicited from experts and the results can be found in \cite{Pollinator_workshop_JAR_paper}.

\begin{figure}[H]
\begin{center}
\includegraphics[width=\textwidth,trim={0cm 2cm 0cm 1cm},clip]{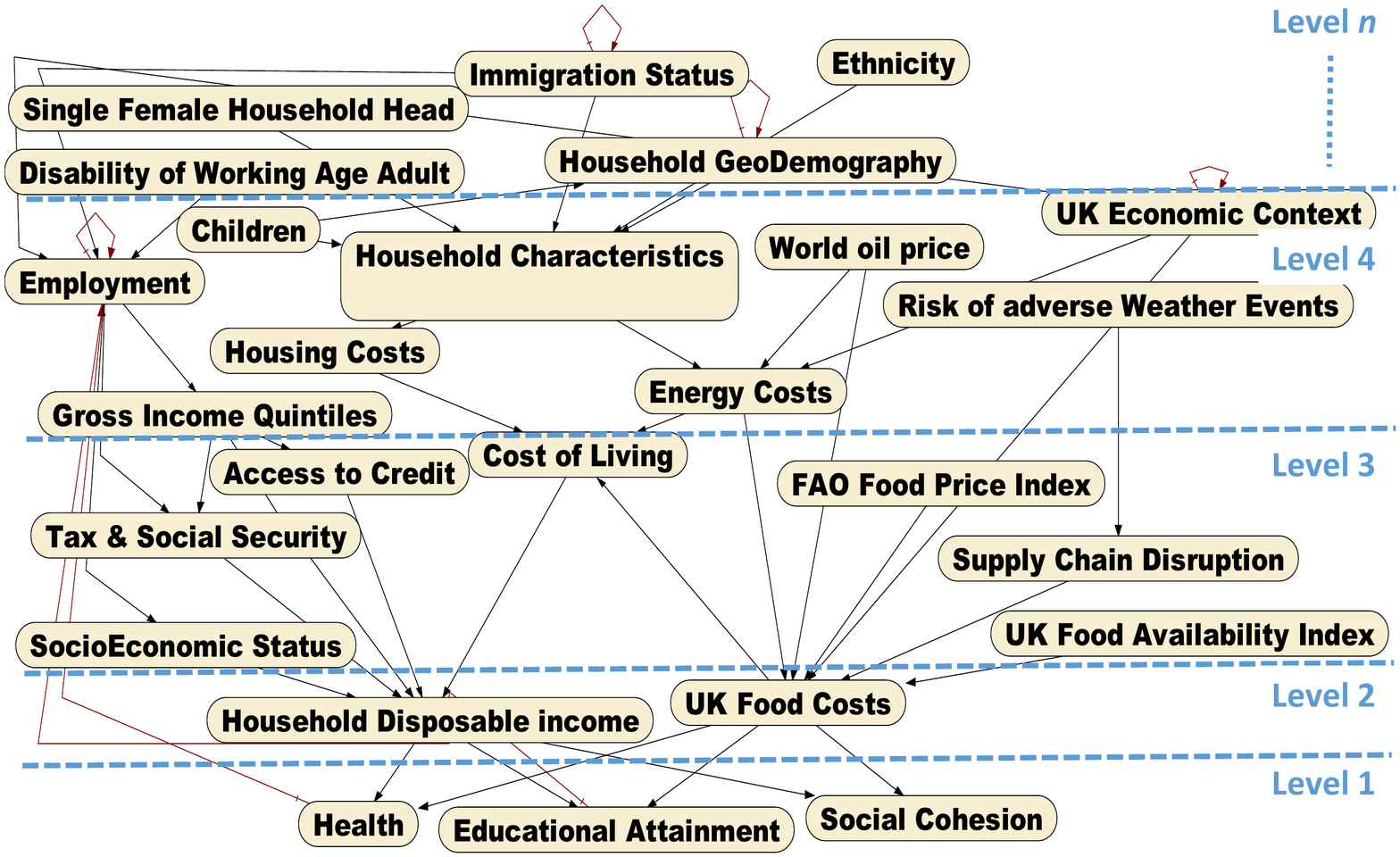}
\caption{Food Security IDSS, red arcs indicating dynamic relationships, from \cite{bookchapter}.}
\label{IDSS_levels}
\end{center}
\end{figure}

\begin{figure}[H]
\begin{center}
\includegraphics[width=0.8\textwidth]{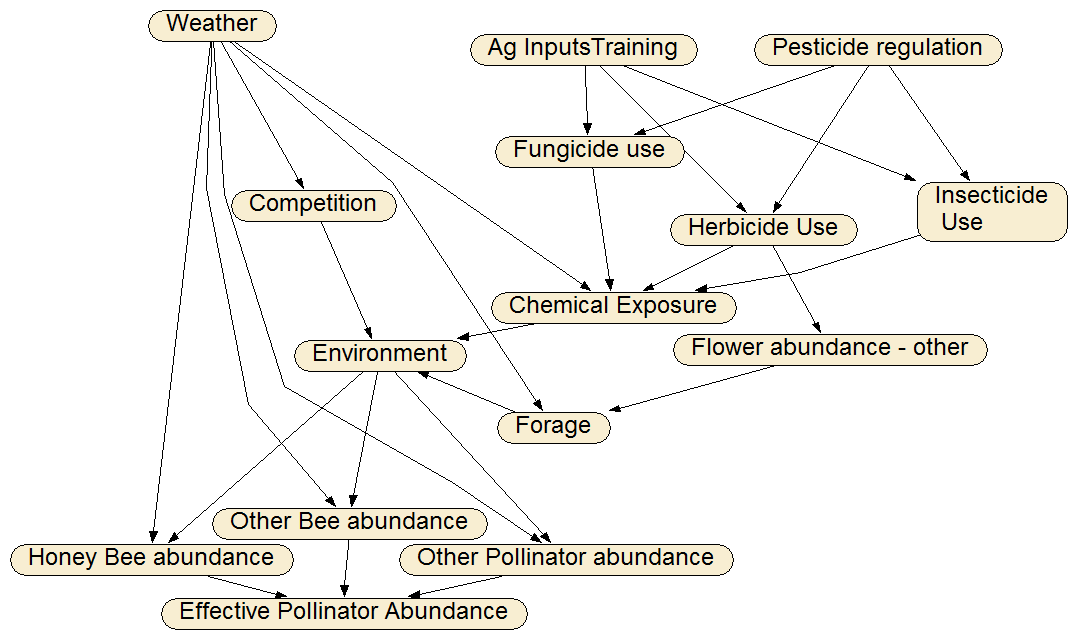}
\caption{Fragment of the pollinator abundance BN sub-subnetwork, from \cite{bookchapter}.}
\label{Pollinator_BN}
\end{center}
\end{figure}

\subsubsection{An Ecological Demonstration}
\label{NativeFishExample}
To illustrate our approach we also use a well known ecological BN called the ``Native Fish" example as introduced in \cite{NativeFish} and discussed further in \cite{FloresNicholson}.  This BN was designed specifically for demonstration purposes, notably introducing non-statisticians to BNs, and is therefore simplified version of a much more complicated process.  However, because the meaning of its variables are transparent and its topology (Version 2 of this model) is just large enough to demonstrate our arguments, this DSS is ideal for illustrating some of our methods. 

This ecological BN is used to model the impact on native fish abundance of pesticide usage on surrounding fields as well as levels of rainfall.  The structure of the BN is given in Figure \ref{NativeFish_BN}.  Our target node is `Native Fish Abundance'.

\begin{figure}[H]
\begin{center}
\includegraphics[scale=0.60,trim={1cm 16cm 4cm 1cm},clip]{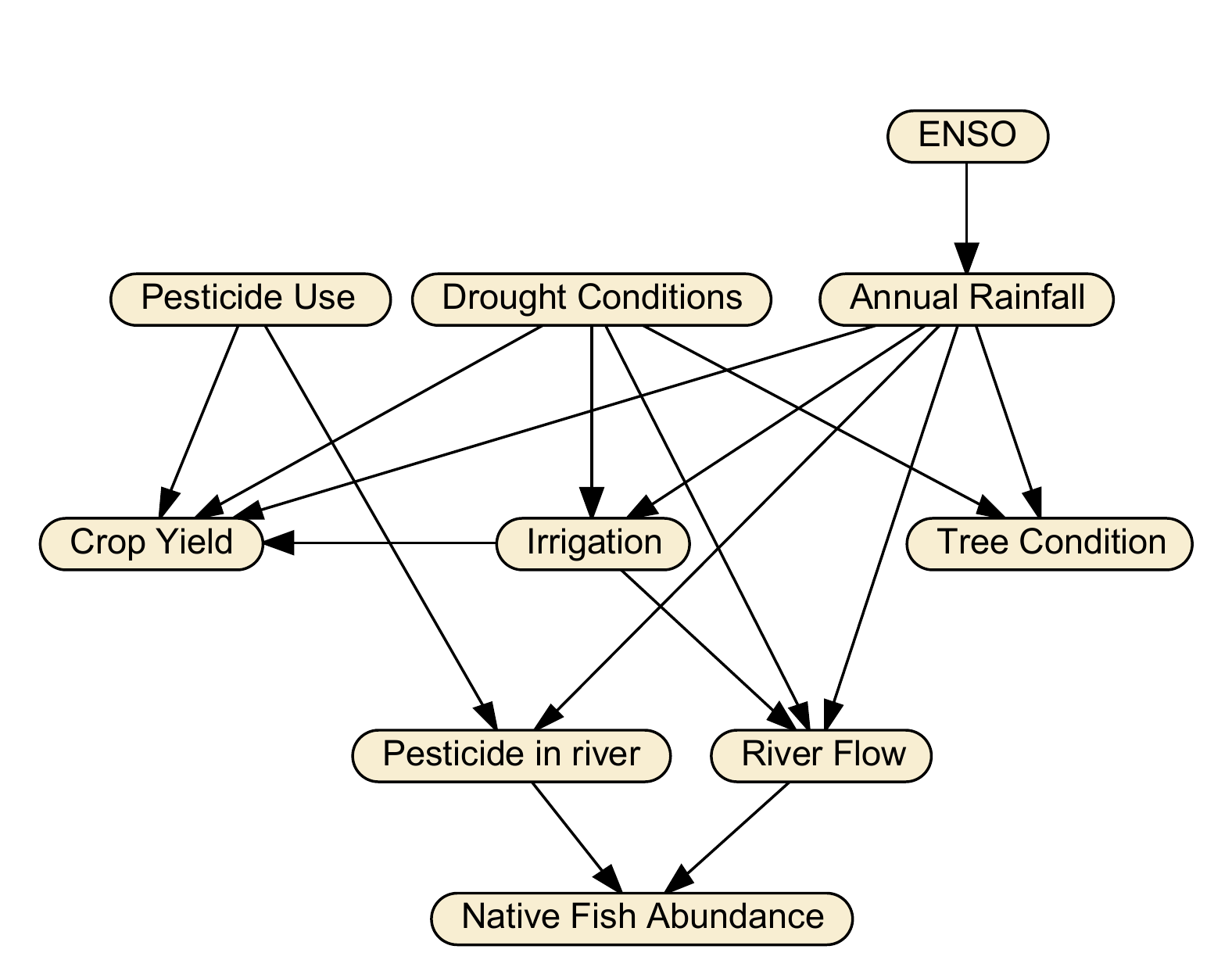}
\caption{``Native Fish" Version 2 BN structure, from \cite{NativeFish}.}
\label{NativeFish_BN}
\end{center}
\end{figure}

\section{Properties of Total Variation Distance for BNs}
\label{TVdist_section}
We begin by outlining the total variation distance, highlighting some of its useful properties which we can directly apply to this robustness analyses. 

Assume $\boldsymbol{X}\triangleq \left( X_1,X_2,\ldots ,X_m\right)$ is a vector of finite discrete random variables taking values $\boldsymbol{x}=\left( x_1,x_2,\ldots ,x_m \right) \in \mathbb{X}_1\times \mathbb{X}_2 \times \cdots \times \mathbb{X}_m$. Let $\boldsymbol{X}_A$, taking values $\boldsymbol{x}_A\in $ $\mathbb{X}_{i_1}\times \mathbb{X}_{i_2}\times \cdots \times \mathbb{X}_{i_{r(A)}}\triangleq \mathbb{X}_A$, denote the subvector of $\boldsymbol{X}$ comprising those components with indices $i\in A$, where $A=\left\{ i_1,i_2,\ldots ,i_{r(A)}\right\} $ denotes a subset of $\left\{ 1,2,\ldots ,m\right\} $.  Let $p_A,$ $q_A$ denote a hypothesised and an alternative joint mass function on $\boldsymbol{X}_A$ and $\mathbb{P}_A(E), \mathbb{Q}_A(E)$ denote the probability with respect to the mass functions $p_A, q_A$ respectively of the set $E=E_{i_1}\times E_{i_2}\times \cdots \times E_{i_{r(A)}}$ where $E_{i_j}\subseteq \mathbb{X}_{i_j}$, $j=1,2,\ldots r(A)$.  Nearly all inferential methodology and certainly all robustness analyses focus on properties of such events \cite[][]{jim_book}. 

\begin{definition}
The \emph{(Total) Variation distance}, $d_V(p_A,q_A)$, is defined in the discrete casee by
\begin{equation*}
d_V(p_A,q_A)\triangleq \frac{1}{2} \sum_{\boldsymbol{x}_A \in \mathbb{X}_A}\left\vert p_A-q_A\right\vert 
\end{equation*}
\end{definition}

\subsection{Variation under Marginalisation and Conditioning}
Measures of variation distance can be applied directly to CPTs.  In this section we define some new objects which will be especially useful in our later development. 

Let $P$ and $Q$, with rows $\boldsymbol{p}_i, \boldsymbol{q}_i$
respectively for $i=1,2,\ldots ,n$, be two $n\times n^{\prime }$ CPT matrices of a random vector $\boldsymbol{Y}$, taking $n^{\prime }$ levels given another random vector $\boldsymbol{X}$, taking $n$ levels.  For a BN, $\boldsymbol{Y}$ will typically be a random variable whilst $\boldsymbol{X}$ will be the vector of its parents; nevertheless when studying junction trees it is also helpful to consider cases when $\boldsymbol{Y}$ is a vector.

There is a natural variation distance we can now define between $P$ and $Q$:

\begin{definition}
Let the variation distance $d_V^{+}\left( P,Q\right)$ between conditional probability tables $P$ and $Q$ be defined by 
\begin{align*}
d_V^{+}\left( P,Q\right) \triangleq \max_{1\leq i\leq n}d_V(\boldsymbol{p}_i,\boldsymbol{q}_i). 
\end{align*}
\end{definition}

\begin{example}
\label{example_TreeCondition_CPT}
Assume that the CPT in \cite{NativeFish}, represented by the transition matrix $P$ below, gives the elicited combined matrix of a panel of experts using a standard protocol \cite[see][for example]{EFSA}. Suppose expert A's individually elicited elicited probabilities are given by the matrix Q:
\begin{table}[H]
\begin{centering}
\begin{tabular}{|c|c|K{1.7cm}|K{1.7cm}|K{1.7cm}|} 
\hline
\multirow{2}{*}{Drought Conditions} & \multirow{2}{*}{Annual Rainfall} & \multicolumn{3}{c|}{$P(Tree Condition|Drought,Rainfall)$} \\
\cline{3-5} %this adds horiztonal line under cells 3 to 5 only
 &  & Good & Damaged & Dead \\
\hline
yes & below average & 0.2 & 0.6 & 0.2 \\
\hline
yes & average & 0.25 & 0.6 & 0.15 \\
\hline
yes & above average & 0.3 & 0.6 & 0.1 \\
\hline
no & below average & 0.7 & 0.25 & 0.05 \\
\hline
no & average & 0.8 & 0.18 & 0.02 \\
\hline
no & above average & 0.9 & 0.09 & 0.01 \\
\hline
\end{tabular}
\end{centering}
\end{table}
We can simplify this to matrix form, denoted by $P$.  Let us assume that this CPT was elicited from experts who disagree on a couple of probabilities resulting in an alternate CPT, $Q$:
\begin{equation*}
P = \begin{pmatrix}
0.20 & 0.60 & 0.20 \\
0.25 & 0.60 & 0.15 \\
0.30 & 0.60 & 0.10 \\
0.70 & 0.25 & 0.05 \\
0.80 & 0.18 & 0.02 \\
0.90 & 0.09 & 0.01
\end{pmatrix}, \hspace{1cm}
Q = \begin{pmatrix}
0.20 & 0.60 & 0.20 \\
0.30 & 0.50 & 0.20 \\
0.30 & 0.60 & 0.10 \\
0.65 & 0.25 & 0.10 \\
0.80 & 0.18 & 0.02 \\
0.90 & 0.10 & 0.00
\end{pmatrix}.
\end{equation*}
We can now compute $d_V^{+}\left( P,Q\right)=\max\{0, 0.10, 0, 0.05, 0, 0.01 \}=0.1$.  Expert A will often be concerned that the substitution of $P$ for $Q$ will not effect significantly the conclusions about the target variables of this panel.  We show below how we can use the diameter calculation to directly measure, in a formal sense, the extent of this substitution and thus allay Expert A's fears that the panel's judgement might be substantially at variance with their own. 
\end{example}

Note that if $\boldsymbol{\rho }(P)$ and $\boldsymbol{\rho }(Q)$ are the vectors of marginal mass functions of $\boldsymbol{Y}$ and $\boldsymbol{\pi }$ is a margin on $\boldsymbol{X}$ then for all possible margins $\boldsymbol{\pi }$ 
\begin{align*}
d_V\left( \boldsymbol{\rho }(P),\boldsymbol{\rho }(Q)\right) \leq d_V^{+}\left( P,Q\right), 
\end{align*}
where $d_V\left( \boldsymbol{\rho }(P),\boldsymbol{\rho }(Q)\right) =d_V^{+}\left( P,Q\right)$ whenever $\boldsymbol{\pi }$ puts all its mass on atoms indexed by $i^{+}$ where
\begin{align*}
i^{+}\triangleq \arg \max_{1\leq i\leq n}d_{V}(\boldsymbol{p}_i,\boldsymbol{q}_i).
\end{align*}
Thus we have that for all possible margins $\boldsymbol{\pi }$ 
\begin{align*}
d_V\left( \boldsymbol{\rho }(P),\boldsymbol{\rho }(Q)\right) \leq d_V^{+}\left( P,Q\right).
\end{align*}
This therefore gives rather coarse, but quick bounds which require only comparisons of the pairs of individual rows of the perturbed CPT. 

Earlier we highlighted that when eliciting a BN we first elicit hypotheses of conditional independence.  Only then do we expand this with a full probability specification through the numerical values in its CPTs.  So we next consider robustness measures associated with small deviations from conditional independence.  The definition we present below is, to our knowledge, a new construction using variation distance on CPTs to determine the measure of dependence between variables.
\begin{definition}
The \emph{diameter}, $d^{+}(P),$ and the \emph{I-local diameter} $d^{I+}(P)$ of a stochastic matrix $P\triangleq
\left\{ p_{ij}\right\} $ are respectively defined as 
\begin{eqnarray*}
d^{+}(P) &=& \frac{1}{2} \max_{1\leq i,i^{\prime }\leq n}\left\{ \sum_{j=1}^{n^{\prime }}\left\vert p_{ij}-p_{i^{\prime }j}\right\vert \right\}, \\
d^{I+}(P) &=& \frac{1}{2} \max_{i,i^{\prime }\in I}\left\{ \sum_{j=1}^{n^{\prime }}\left\vert p_{ij}-p_{i^{\prime }j}\right\vert \right\}.
\end{eqnarray*}
\end{definition}
\begin{example}
The values of diameters (typical of those found in many exercises) for each of the CPTs of the Native Fish BN from \cite{FloresNicholson}, together with those obtained in an elicitation exercise associated with the pollinator example \cite{Pollinator_workshop_JAR_paper} are given in the tables below. Discrepancies passing through CPTS with diameters close to $1$ might be retained as different target distributions. However, once discrepancies pass through more than two CPTs with diameters less than $0.7$ these usually 	quickly dissolve, for reasons we discuss below. 

\begin{figure}[H]
\begin{center}
\begin{tabular}{ | l | c | l | c|}
    \hline
    \textbf{Node} & \textbf{Diameter} & \textbf{Node} & \textbf{Diameter} \\ \hline
    Annual Rainfall & 0.65 & Crop Yield & 0.98 \\ \hline
    River Flow & 0.98 & Irrigation & 0.94 \\ \hline
    Pesticide in River & 0.7 &  Tree Condition & 0.7 \\ \hline
    Native Fish Abundance & 0.84 & & \\ 
    \hline
\end{tabular}
\caption{Diameters of each CPT in the Native Fish BN.}
\end{center}
\end{figure}

\begin{figure}[H]
\begin{center}
\begin{tabular}{ | l | c |}
    \hline
    Honey Bee Abundance & 0.66 \\ \hline
    Other Bee Abundance & 0.55 \\ \hline
    Other Pollinator Abundance & 0.54 \\ 
    \hline
\end{tabular}
\caption{Diameters of each CPT in Pollinator sub-network BN.}
\end{center}
\end{figure}
\end{example}

The size of the diameter of a conditional probability table $P$ is a measure of the dependence of $\boldsymbol{Y}$ on $\boldsymbol{X}$.  This is because whenever $\boldsymbol{Y}\amalg \boldsymbol{X}$ all rows of $P$ will be equal and so $d^{+}(P)=0$.  It is easy to check that whenever some non-trivial function $\tau (\boldsymbol{Y)}$ of $\boldsymbol{Y}$ can be written as a deterministic function of $\boldsymbol{X}$ then $d^{+}(P)=1$, its maximum value.  So when there is only a weak relationship between $\boldsymbol{Y}$ and $\boldsymbol{X}$, in the sense that changing the different levels of $\boldsymbol{X}$ impacts only slightly on the conditional mass function of $\boldsymbol{Y}$, then $d^{+}(P)\bumpeq 0$.  Note that unless $P$ is symmetric, the diameter of $\boldsymbol{Y}$ on $\boldsymbol{X}$ is not the same as the diameter of $\boldsymbol{X}$ on $\boldsymbol{Y}$, in fact the difference between these can be arbitrarily close to $1$ \cite[see][]{skw_thesis}.

The $I$-local diameter has the same property, where this time it is  conditional on $\boldsymbol{X}$ taking values only in the set of levels $I$.  This is useful when comparing the efficacy of deleting a parent in a BN or when combining a collection of rows of the CPT/levels of $\boldsymbol{X}$ into a single entry: see below. 

\subsection{Variation and Mixtures}

\subsubsection{Approximations associated with mixing}
A useful and well-known property of total variation is its convexity under mixing in the following sense. Let $\boldsymbol{\pi }=(\pi _1,\pi _2,\ldots ,\pi _n)$, $\boldsymbol{\pi }^{\prime }=(\pi _1^{\prime },\pi _2^{\prime },\ldots ,\pi _{n^{\prime }}^{\prime })$, $\boldsymbol{q}_i=(q_{i1},q_{i2},\ldots ,q_{in})$, $\boldsymbol{p}_i=(p_{i1},p_{i2},\ldots ,p_{in^{\prime }})$ and define
\begin{align*}
\boldsymbol{q}_{\boldsymbol{\pi }}\triangleq \sum_{i=1}^{n}\pi _{i}\boldsymbol{q}_i, \hspace{1cm} \boldsymbol{p}_{\boldsymbol{\pi }^{\prime }}\triangleq \sum_{i=1}^{n^{\prime }}\pi _i^{\prime }\boldsymbol{p}_i,
\end{align*}
then 

\begin{lemma}
\label{mixing_lemma}
\begin{equation*}
d_V(\boldsymbol{p}_{\boldsymbol{\pi }^{\prime }},\boldsymbol{q}_{\boldsymbol{\pi }})\leq \sum_{i=1}^{n}\sum_{i^{\prime }=1}^{n^{\prime }}\pi_i\pi _{i^{\prime }}^{\prime }d_V(\boldsymbol{p}_{i^{\prime }},\boldsymbol{q}_i).  \label{convexity of TV}
\end{equation*}
\end{lemma}

\begin{proof}
See Appendix \ref{appendixA1}.
\end{proof}

In particular, if we know extremal distributions are small then so are convex linear combinations of these.  Such processes occur for example in the calculation of a  margin: here of a target variable. This enables us to prove a number of useful results concerning the contraction of error under learning in a BN: see below. 

Combining our new definitions of diameter with variation distance we prove the following result that enables us to track this distance through a given BN:

\begin{theorem}
\label{theorem_totvar_diameter}
Let $\boldsymbol{\pi }_1$ and $\boldsymbol{\pi }_2$ be two possible margins of vectors $\boldsymbol{X}$ and $\boldsymbol{Y}$ of random variables and suppose that $P(\boldsymbol{Y|X})$ is the (shared) CPT of the concatenated levels of the conditional $\boldsymbol{Y}|\boldsymbol{X}$ and
that $\boldsymbol{\rho }_{1}$ and $\boldsymbol{\rho }_{2}$ are the margins of $Y$. Then 
\begin{align*}
d_V(\boldsymbol{\rho }_1,\boldsymbol{\rho }_2)\leq d^{+}(P(\boldsymbol{Y|X}))d_V(\boldsymbol{\pi }_1,\boldsymbol{\pi }_2). 
\end{align*}
\end{theorem}

\begin{proof}
See Appendix \ref{appendixA2}.
\end{proof}

This property will be exploited below in the study of BNs.  Note for example that if $P(\boldsymbol{Y|X})$ has been specified accurately, but that the margin\ $\boldsymbol{\pi }_1$ is uncertain, then our marginal beliefs about $\boldsymbol{Y}$ are no more uncertain than those about $\boldsymbol{X}$, because by definition $d^{+}(P(\boldsymbol{Y|X}))\leq 1$.  More importantly we have a bound on how much our uncertainty, quantified in terms of total variation, reduces in terms of $d^{+}(P(\boldsymbol{Y|X}))$ -- a measure of how far away $\boldsymbol{Y}$ is from independence of $\boldsymbol{X}$. 

\begin{example}
Let us once again look at the CPT of `Tree Condition', $P$, which had a binary parent `Drought' and a three-state parent `Rainfall'.  The joint distribution can be calculated from CPTs as $\boldsymbol{\pi}_1=(0.05,0.175,0.025,0.15,0.525,0.075)$.  Suppose another expert proposed he different probability vector $\boldsymbol{\pi}_2 =(0.05,0.275,0.03,0.15,0.4,0.095)$.  We have previously calculated $d^+(P)=0.7$ and can calculate that $d_V(\boldsymbol{\pi }_1,\boldsymbol{\pi }_2)=0.125$.  Therefore Theorem \ref{theorem_totvar_diameter} gives:
\begin{align*}
d_V(\boldsymbol{\rho }_1,\boldsymbol{\rho }_2) &\leq d^{+}(P(\boldsymbol{Y|X}))d_V(\boldsymbol{\pi }_1,\boldsymbol{\pi }_2) = 0.7 \times 0.125 = 0.0875
\end{align*}
However, we can of course calculate this margin exactly as $d_V(\boldsymbol{\rho}_1,\boldsymbol{\rho}_2) = 0.0555$.  However, is we knew only the extreme entries of $P$ then we could still calculate our bound which is of the right order of magnitude: a property we have found to be typical of the types of CPTs we habitually elicit. 
\end{example}

\subsubsection{A Global Bound Approximation}
There is another bound which applies when not only a margin $\boldsymbol{\pi }_1$ of $\boldsymbol{X}$ is perturbed to $\boldsymbol{\pi }_2$, but also the conditional mass functions of $\boldsymbol{Y|X}$ is simultaneously perturbed.  Occasionally we need variation bounds on the consequent perturbation on the margins $\boldsymbol{\rho }_1,\boldsymbol{\rho }_2$ of $\boldsymbol{Y}$:

\begin{definition}
Let the superbound, $d_V^{\ast }\left( P,Q\right)$, between stochastic matrices $P$ and $Q$ be defined by 
\begin{align*}
d_V^{\ast }\left( P,Q\right) \triangleq \max_{1\leq i,i^{\prime }\leq n}d_V(\boldsymbol{p}_i,\boldsymbol{q}_{i^{\prime }})\leq 1.
\end{align*}
\end{definition}

So here we compare variation distances between each row of $P$ and possibly different rows of $Q$ before selecting the largest difference. Note by definition and the triangle inequality that 
\begin{equation}
d_V^{+}\left( P,Q\right) \leq d_V^{\ast }\left( P,Q\right) \leq \max \left\{ d_V^{+}\left( P,Q\right) +\max \left\{ d(P),d(Q)\right\} ,1\right\}. \label{oldinequality}
\end{equation}

\begin{example}
Let us compare the two alternative CPTs, $P$ and $Q$, for the `Tree Condition' node as introduced in Example \ref{example_TreeCondition_CPT}.  The value of $d_V^*(P,Q)$ can be calculated directly from the total variation distance between every possible pairwise combination of rows in $P$ and $Q$.  For this example $d_V^*(P,Q)=0.7$ corresponding to $d_V(\boldsymbol{p}_1,\boldsymbol{q}_6)=d_V(\boldsymbol{p}_6,\boldsymbol{q}_1)$.
\end{example}

Let $P_{A|B}$, $Q_{A|B}$ represent respectively the conditional probability mass functions of $\boldsymbol{X}_A|$ $\boldsymbol{X}_B$ under the hypothesis and alternative given $\boldsymbol{X}_B$, where without loss we can assume that $A$ and $B$ are disjoint. Notice that these can be seen as CPTs whose rows correspond to the different values of the vector $\boldsymbol{x}_B$. Then under our definitions of transition matrices above whenever $\boldsymbol{X}_A\amalg \boldsymbol{X}_B$
\begin{align*}
d_V^{+}(P_{A|B},Q_{A|B})=d_V^{\ast
}(P_{A|B},Q_{A|B})=d_V(p_{A},q_{A}).
\end{align*}
This arises simply because $\boldsymbol{X}_A\amalg \boldsymbol{X}_B$ implies that all rows in the CPT matrix are equal to each other and so equal to the corresponding margin on $\boldsymbol{X}_A$.  Thus we see that standard analyses that elicit irrelevances or independences translate here into equations on variation distance.  We will see later that this enables us to study the implications of models where the embedded conditional independences are only approximately true. 

\begin{definition}
The \emph{stochastic variation matrix} $D^{+}(P)=\{D^{+}(i,j)\}_{1\leq i,j\leq n}$ is the $n\times n$ symmetric matrix whose entries are the variation distances between the different rows of the matrix $P$.
\end{definition}

We will later use this construction to draw out useful functions of the explanatory variables associated with a particular variable of focus.

Now note that we can write 
\begin{equation*}
\boldsymbol{\pi }_1 =(1-\beta )\boldsymbol{\pi }_1^* +\beta  \boldsymbol{\pi }_{1\wedge 2}, \hspace{1cm}
\boldsymbol{\pi }_2 =(1-\beta )\boldsymbol{\pi }_2^* +\beta  \boldsymbol{\pi }_{1\wedge 2},
\end{equation*}
where $(1-\beta )=d_V\left( \boldsymbol{\pi }_1,\boldsymbol{\pi }_2\right) $ and where without loss we can assume the mixing process is shared by the two mass functions, so points are drawn either from $\boldsymbol{\pi }_{1\wedge 2}$ or alternatively something drawn from either $\boldsymbol{\pi }_1^*$ or $\boldsymbol{\pi }_2^*$ (see Supplementary Material for more a more detailed construction).  Using the same argument as for when $P_1=P_2$ 
\begin{eqnarray*}
d_V\left( \boldsymbol{\rho }_1,\boldsymbol{\rho }_2\right)
&=&d_V\left( \boldsymbol{\pi }_1 P_1,\boldsymbol{\pi }_2 P_2 \right)
\\
&=&d_V\left( \left( (1-\beta )\boldsymbol{\pi }_1^* +\beta \boldsymbol{\pi }_{1\wedge 2}\right) P_1,\left( (1-\beta )\boldsymbol{\pi }_2^* +\beta \boldsymbol{\pi }_{1\wedge 2}\right) P_2\right) \\
&\leq &\beta d_V\left( \boldsymbol{\pi }_{1\wedge 2}P_1,\boldsymbol{\pi }_{1\wedge 2}P_2\right) +(1-\beta )d_V\left( \boldsymbol{\pi }_1^* P_1,\boldsymbol{\pi }_2^* P_2\right) \\
&\leq &\beta d_V^{+}\left( P_1,P_2\right) +(1-\beta )d_V^{\ast
}\left( P_1,P_2\right).
\end{eqnarray*}
We can then show
\begin{equation*}
d_V\left( \boldsymbol{\rho }_1,\boldsymbol{\rho }_2\right) \leq
d_V^{+}\left( P_1,P_2\right) +d_V\left( \boldsymbol{\pi }_1,\boldsymbol{\pi }_2\right) d_V^{\ast }\left( P_1,P_2\right).
\label{First iteration1}
\end{equation*}
Using Equation \ref{oldinequality} in particular we have that 
\begin{equation*}
d_V\left( \boldsymbol{\rho }_1,\boldsymbol{\rho }_2\right) \leq
\left\{ 1+d_V\left( \boldsymbol{\pi }_1,\boldsymbol{\pi }_2\right) \right\} d_V^{+}\left( P_1,P_2\right) +d_V\left( \boldsymbol{\pi }_1,\boldsymbol{\pi }_2\right) \max \left\{ d(P_1),d(P_2)\right\}.
\label{first iteration 2}
\end{equation*}

\section{Approximations of the CPTs in a known BN}
\label{approximations_section}
Suppose all clients are content that the conditional independences in a given BN are valid.  Without changing the random variables in the system we are now interested in finding ways of approximating the graphical model and refining initial probability estimates within this given BN.

\subsection{Diameter Bounds when Marginalising or Conditioning}
We now present some basic results about diameters of the transition matrices between two vectors of random variables under various marginalisations and conditioning of the subvectors.  These bounds are particularly helpful when moving from a BN to a junction tree.

Let $\boldsymbol{X}=\left( \boldsymbol{X}_{1},\boldsymbol{X}_{2}\right) ,$ $\boldsymbol{Y}=\left( \boldsymbol{Y}_{1},\boldsymbol{Y}_{2}\right) $ and $P_{\boldsymbol{Y|X}}$ $(P_{\boldsymbol{Y|X}_{1}})$ be, respectively, the transition matrix associated with the conditional distribution of $\boldsymbol{Y|X}$ (the same conditional distribution $\boldsymbol{Y|X}_{1}$ but now with $\boldsymbol{X}_{2}$ marginalised out).  Let $d^+ (P_{\boldsymbol{Y|X}}),(d^+ (P_{\boldsymbol{Y|X}_{1}}))$ denote their respective diameters.

\begin{lemma}
\label{diameter_marginalisation}
\begin{equation*}
d^+ (P_{\boldsymbol{Y|X}_{1}})\leq d^+ (P_{\boldsymbol{Y|X}}).
\end{equation*}
\end{lemma}

\begin{proof}
This is immediate since each of the rows of $P_{\boldsymbol{Y|X}_{1}}$ is a weighted average (the weights on row labelled $\boldsymbol{x}_{1} $ corresponding to the masses on $\boldsymbol{X}_{2}|\boldsymbol{X}_{1}=\boldsymbol{x}_{1}$).
\end{proof} 

Note that this bound is tight in the sense that it is attained for a particular distribution on $\boldsymbol{X}_{2}|\boldsymbol{X}_{1}=
\boldsymbol{x}_{1}$.  Suppose $d^+(P_{\boldsymbol{Y|X}})$ is attained when we compare the row $\left( \boldsymbol{x}_{1},\boldsymbol{x}_{2}\right) $ with $\left( \boldsymbol{x}_{1}^{\prime },\boldsymbol{x}_{2}^{\prime }\right) $ and 
\begin{equation*}
P\left( \boldsymbol{X}_{2}=\boldsymbol{x}_{2}|\boldsymbol{X}_{1}=\boldsymbol{x}_{1}\right) =1\text{ and }P\left( \boldsymbol{X}_{2}=\boldsymbol{x}_{2}^{\prime }|\boldsymbol{X}_{1}=\boldsymbol{x}_{1}^{\prime }\right) =1,
\end{equation*}
then it is easy to check that $d^+(P_{\boldsymbol{Y|X}_{1}})\leq d^+(P_{\boldsymbol{Y|X}})$.

\begin{lemma}
\label{lemma14}
Using the obvious notation, for any two joint probability mass functions $p_{\boldsymbol{X,Y}}(\boldsymbol{x},\boldsymbol{y}), p_{\boldsymbol{X,Y}
}^{\prime }(\boldsymbol{x},\boldsymbol{y})$ over $\boldsymbol{X,Y}$ 
\begin{equation*}
d_{V}(p_{\boldsymbol{X,Y}}(\boldsymbol{x},\boldsymbol{y}),p_{\boldsymbol{X,Y}}^{\prime }(\boldsymbol{x},\boldsymbol{y}))\leq \inf \left\{ d_{V}(p_{\boldsymbol{X}}(\boldsymbol{x}),p_{\boldsymbol{X}}^{\prime }(\boldsymbol{x}))+\sup_{\boldsymbol{x}}d_{V}(p_{\boldsymbol{Y}|\boldsymbol{X}}(\boldsymbol{y
}|\boldsymbol{x}),p_{\boldsymbol{Y}|\boldsymbol{X}}^{\prime }(\boldsymbol{y}|\boldsymbol{x})),1\right\}.
\end{equation*}
\end{lemma}

\begin{proof}
See Appendix \ref{appendixB1}.
\end{proof}

Finally, we can determine a bound on the diameter of a CPT in which many variables are dependent on the same set.  This will often be the case when we are looking at a simple path of a junction tree in which a separator contains more than one variable:

\begin{lemma}
\label{diameter_nonsingle_separator}
$d^+(P_{\boldsymbol{Y|X}})\leq \inf \left\{ d^+(P_{\boldsymbol{Y}_{1}\boldsymbol{|X}})+ d^+(P_{\boldsymbol{Y}_{2}\boldsymbol{|X,Y}_{1}}),1\right\}. $
\end{lemma}

\begin{proof}
See Appendix \ref{appendixB2}.
\end{proof}

These results may seem trivial, however they enable us to bound the diameters of CPTs in our junction tree path, using the diameters already calculated from the original CPTs in the BN.  This enables us to study the robustness to misspecification without calculating any new information.

\subsection{Diminishing tree propagated approximation error}
The following result explains why when using standard propagation algorithms on updating one of the clique margins $C_1$, the knock on effect on the other clique margins becomes weaker and weaker as the updated cliques become progressively more remote from $C_1$ - a property \cite{albrecht_nicholson_whittle} exploit in their work.  Furthermore the extent of the deviation can be measured, in the sense that it can be bounded above. This enables us to bound the potential extent of error in the distributions of focus variables induced from the misspecification of structure or various CPTs in the BN.  This is particularly useful when we elicit a large BN and want to know how far away from target nodes we need to elicit the corresponding CPTs accurately. 

\begin{theorem}\label{clique_thm}
Let $\left( C_1,C_2,C_3,\ldots ,C_k\right)$, from $C_1$ to $C_k$, be the minimal sequence of cliques with associated separators $\left( S_2,S_3,\ldots,S_k\right)$. Let each undirected edge of the marginalised junction tree be denoted by $\delta_i$ for $i=1,2,\hdots,k$; the diameter of the conditional probability table between the two sequential nodes, for example $\delta_1=d^+(P(S_2|C_1)), \delta_2=d^+(P(S_3|S_2)), \hdots, \delta_k=d^+(P(C_k|S_k))$. Then
\begin{align*}
d_V(p_{C_k}(\boldsymbol{x}_{C_k}),q_{C_k}(\boldsymbol{x}_{C_k}))\leq d_V(p_{C_1}(\boldsymbol{x}_{C_1}),q_{C_1}(\boldsymbol{x}_{C_1}))  \prod\limits_{i=1}^{k} \delta _i.
\end{align*}
\end{theorem}

\begin{proof}
By Lemma 2.1 we can rewrite our junction tree to marginalise over internal cliques leaving us with the graphical structure:
\begin{figure}[H]
\begin{center}
\scalebox{0.75}{
\begin{tikzpicture}[-,shorten >=1pt,auto,node distance=2cm, thick]
\tikzstyle{state}=[fill=white,draw,text=black,shape=circle] 
\tikzstyle{sq.state}=[fill=white,draw,text=black,shape=rectangle] 
\node[state](C1)  [] {$C_1$};
\node[sq.state](S2)  [right=1cm of C1]{$S_2$};
\node[sq.state](S3)  [right=1cm of S2]{$S_3$};
\node[state](invisible)  [right=1cm of S3,draw=white] {...};
\node[sq.state](Sk)  [right=1cm of invisible]{$S_k$};
\node[state](Ck)  [right=1cm of Sk] {$C_k$};
\path (C1) edge (S2);
\path (S2) edge (S3);
\path (S3) edge (invisible);
\path (invisible) edge (Sk);
\path (Sk) edge (Ck);
(
\end{tikzpicture}
}
\end{center}
\end{figure}
Let each undirected edge be denoted by $\delta_i$ for $i=1,2,\hdots,k$; the diameter of the conditional probability table between the two sequential nodes. Giving $\delta_1=d^+(P(S_2|C_1)), \delta_2=d^+(P(S_3|S_2)), \hdots, \delta_k=d^+(P(C_k|S_k))$.
By successive application of Theorem 3.7:
\begin{align*}
d_V(p_{C_k}(\boldsymbol{x}_{C_k})&,q_{C_k}(\boldsymbol{x}_{C_k})) \leq d^+(P(C_k|S_k))d_V(p_{S_k}(\boldsymbol{x}_{S_k}),q_{S_k}(\boldsymbol{x}_{S_k})) \\
&\leq d^+(P(C_k|S_k)) d^+(P(S_k|S_{k-1})) d_V(p_{S_{k-1}}(\boldsymbol{x}_{S_{k-1}}),q_{S_{k-1}}(\boldsymbol{x}_{S_{k-1}})) \\
&\leq d^+(P(C_k|S_k)) d^+(P(S_k|S_{k-1})) \hdots d^+(P(S_3|S_2)) d^+(P(S_2|C_1)) d_V(p_{C_1}(\boldsymbol{x}_{C_1}),q_{C_1}(\boldsymbol{x}_{C_1})) \\
&= \left( \prod_{i=1}^{k} \delta_i \right) d_V(p_{C_1}(\boldsymbol{x}_{C_1}),q_{C_1}(\boldsymbol{x}_{C_1}))
\end{align*}
\end{proof}
Next we define the impact of one clique upon another in order to ascertain the diminishing effect of errors downstream in the causal chain.

\begin{definition}
Define the \emph{impact }$I(C_k|C_1)$ of $C_1$ on $C_k$ to be $\prod\limits_{i=1}^{k}\delta _{i}$.
\end{definition}

The impact of one clique on another is a simple measure of the maximum possible influence the misspecification of one set of clique probabilities could have on another as measured by a bound on the variation distance.  Note that in general we can label the edges of a junction tree (which are also labelled by a separator between adjacent cliques) $C_i$ and $C_j$ by two diameters $\delta _{i\rightarrow j}$ and $\delta _{j\rightarrow i}$ one measuring the impact of $i$ on $j$ and the other the impact of $j$ on $i$.  Note that these two impacts are not necessarily equal, and are often very different.  However, in the contexts we consider here (where our primary interest concerns the robustness of the margins of an identified subset of attributes) we usually need to focus on propagation in a single direction.  Furthermore, if the BN is constructed consistently with a conjectured causal directionality in mind, then this directionality often tends to have the attributes at the end of the causal chain.  This means that the diameters we need can often be calculated directly from the diameter of the elicited CPTs of the BN.  

\begin{example}
The two simple BNs we have used in our running example are not deep enough to illustrate the usefulness of this result, whilst the full IDSS is far too complicated. So instead we use here a simplification of a BN used to model radicalisation processes one of the authors has elicited, where the precise meaning of the nodes is confidential but not relevant to the points we mean to illustrate.
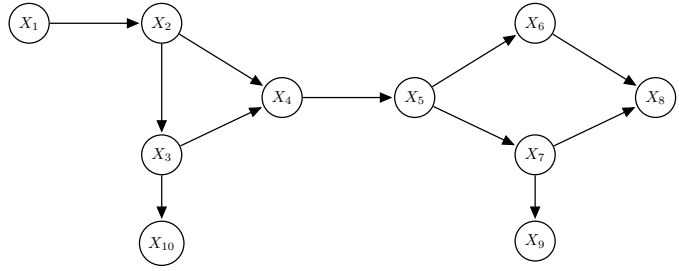
\begin{figure}[H]
\begin{center}
\scalebox{0.6}{
\begin{tikzpicture}[->,>=triangle 45,shorten >=1pt,auto,node distance=2cm, thick]
\tikzstyle{every state}=[fill=white,draw,text=black, shape=circle, rounded corners] 
\node[state](X1)  [draw=black,fill=white, text=black,] {$X_1$};
\node[state](X2)  [right=2cm of X1, draw=black,fill=white]{$X_2$};
\node[state](X3)  [below=2cm of X2, draw=black,fill=white,  text=black]{$X_3$};
\node[state](X4)  [below right=1cm and 2cm of X2, draw=black,fill=white,  text=black]{$X_4$};
\node[state](X5)  [right=2cm of X4, draw=black,fill=white]{$X_5$};
\node[state](X6)  [above right=1cm and 2cm of X5, draw=black,fill=white]{$X_6$};
\node[state](X7)  [below= 2cm of X6, draw=black,fill=white]{$X_7$};
\node[state](X8)  [below right=1cm and 2cm of X6, draw=black,fill=white]{$X_8$};
\node[state](X9)  [below=1cm of X7, draw=black,fill=white,  text=black]{$X_9$};
\node[state](X10)  [below=1cm of X3, draw=black,fill=white,  text=black]{$X_{10}$};
\path (X1) edge (X2);
\path (X3) edge (X4);
\path (X2) edge (X4);
\path (X2) edge (X3);
\path (X3) edge (X10);
\path (X4) edge (X5);
\path (X5) edge (X6);
\path (X5) edge (X7);
\path (X7) edge (X9);
\path (X6) edge (X8);
\path (X7) edge (X8);
(
\end{tikzpicture}
}
\caption{BN example to determine Impact of cliques.}
\label{JT_example}
\end{center}
\end{figure}

Let us label the cliques to satisfy the running intersection property:
\begin{align*}
\begin{matrix}
C_1 = \{ X_1, X_2 \}, \hspace{0.5cm} C_2 = \{ X_2, X_3, X_4 \}, \hspace{0.5cm} C_3 = \{ X_4, X_5\}, \\
C_4 = \{ X_5, X_6, X_7 \}, \hspace{0.5cm} C_5 =\{ X_6, X_7, X_8 \}, \hspace{0.5cm} C_6=\{ X_3, X_{10} \}, \hspace{0.5cm} C_7 = \{ X_7, X_9 \}
\end{matrix}
\end{align*}
Giving us separators:
\begin{align*}
\begin{matrix}
S_2 = \{ X_2 \}, \hspace{0.5cm} S_3 = \{ X_4 \}, \hspace{0.5cm} S_4 = \{ X_5 \}, \hspace{0.5cm} S_5 = \{ X_6, X_7\} \\
S_6 = C_2 \cap C_6 = \{ X_3 \}, \hspace{0.5cm} S_7 = C_5 \cap C_7 = \{ X_7 \}
\end{matrix}
\end{align*}

Suppose we wish to determine the effect on $X_9$ if we perturb $X_1$. Draw the ancestral graph of nodes $X_1$ and $X_9$, derive the impact formula (which is simply the product of diameters of each separator conditional on the previous previous separators):
\begin{align*}
I(X_9|C_1)&=p(X_2|X_1)p(X_4|X_2)p(X_5|X_4)p(X_7|X_5)p(X_9|X_7) \\
&\leq d^+(X_2)d^+(X_4)d^+(X_5)d^+(X_7)d^+(X_9) 
\end{align*}

Extending this further, we can determine the impact on cliques $X_6$ and $X_7$ simultaneously, if we perturb both $X_1$ and $X_2$. Following the same steps of creating cliques and separators for the ancestral graph of these nodes, the impact is given as:
\begin{align*}
I(X_6, X_7|X_1, X_2)&=p(X_2|X_1)p(X_4|X_2)p(X_5|X_4)p(X_6,X_7|X_5).
\end{align*}
This can be written in terms of the original BN CPTs using Lemma \ref{diameter_nonsingle_separator}, as some separators contain more than one node:
\begin{align*}
I(X_6, X_7|X_1, X_2)&\leq d^+(X_2) d^+(X_4) d^+(X_5) \bigg[ \inf \{ d^+(X_6|X_5) + d^+(X_7|,X_6,X_5), 1 \} \bigg] \\
&\leq d^+(X_2) d^+(X_4) d^+(X_5) \bigg[ \inf \{ d^+(X_6|X_5) + d^+(X_7|,X_5), 1 \} \bigg] 
\end{align*}
 \end{example}

There are various practical corollaries to the simple theorem above:

\begin{corollary}
If $\mathcal{G}$ is decomposable and $C_i$ lies on the minimal sequence between $C_1$ and $C_k$ then if all attributes are in $C_k$ then the probabilities of $C_i$ have higher influence on $C_k$ than those of $C_1.$
\end{corollary}

As we indicated above, these bounds can be applied to \textit{any} BN.  We recommend following the construction below to ensure that your BN is in a suitable format to apply Theorem \ref{clique_thm}: 

\begin{itemize}
\item Begin with a BN $\mathcal{G}$, the diameters of whose CPTs have been
provisionally elicited.
\item Identify a donating variable or complete vector $\boldsymbol{X}_{i}$ of $\mathcal{G}$ and the vector of focus $\boldsymbol{X}_{k}$.
\item Find the ancestral set of $\boldsymbol{X}_{i},\boldsymbol{X}_{k}$ in  $\mathcal{G}$.
\item Construct the ancestral graph, $A$, which has variables $(X_{1},X_{2},\ldots ,X_{n})$ where the order of these vertices are chosen compatible with $\mathcal{G}$.
\item Create a triangularised version, $A^{\ast}$, of $A$ and find its junction tree $J$.  Denote the clique containing $\boldsymbol{X}_{i}$ as $C_{1}$ and the clique containing $\boldsymbol{X}_{k}$, $C_{k}$. 
\item Find the single path $J^{\ast}$ starting from clique $C_{1}$ to $
C_{k}$ labelling the cliques in order $C_{1},C_{2},\ldots ,C_{k}$. 
\item Remove all variables that are not in one of these cliques.
\end{itemize}

Note that these influences provide a very useful tool for prioritisation of the elicitation in a BN.  For example, if we can obtain estimates of influence across a junction tree (either from direct elicitation of $\delta $ or alternatively after having performed a preliminary coarse elicitation of the corresponding CPTs) then we can use these influences to identify which of those CPTs to refine.   For example suppose all attributes consisting of the subvectors of variables of interest lie in a single clique.  We can then follow the simple guidelines:

\begin{itemize}
\item Refine the elicitation of the CPTs whose attributes and parents lie in this clique,
\item Elicit the CPTs associated with parents/separators with the most influence,
\item Use the influence formula (Theorem \ref{clique_thm}) to guide the refinement of the CPTs associated with other parents or parents of parents.
\end{itemize}

\subsection{Approximations associated with a general BN}
In a junction tree each vector has just a single parent within a given compatible ordering.  Of course in the case of a BN this is no longer necessarily true.  We would still like to find the impact bound of one variable on another and so annotate each of its directed edges with a value between zero and one which reflects this.  The result below gives us a way of coding this impact in a useful way. 

Suppose $\boldsymbol{Y}$, taking values $\boldsymbol{y}\in \mathbb{Y}$, is potentially dependent on $k$ vectors $\boldsymbol{X}=\left( \boldsymbol{X}_1,\boldsymbol{X}_2,\ldots \boldsymbol{X}_k\right)$, taking values $\boldsymbol{x}=\left( \boldsymbol{x}_1,\boldsymbol{x}_2,\ldots,\boldsymbol{x}_k\right) \in \mathbb{X=X}_1\times \mathbb{X}_2\times \cdots \times \mathbb{X}_k$.  For $j=1,2,\ldots ,k$ let $\boldsymbol{x}_{\widehat{j}}\in \mathbb{X}_{\widehat{j}}\triangleq \mathbb{X}_{\left\{1,2,\ldots ,k\right\} \backslash \left\{ j\right\} }$ be a vector of values of other variables $\boldsymbol{X}_{\widehat{j}}$.  Let the CPT of $\boldsymbol{Y}$ given $\boldsymbol{X}$ be $P$ so that its diameter is given by 
\begin{align*}
d^{+}(P)= \frac{1}{2} \max_{\boldsymbol{x,x}^{\prime }\in \mathbb{X}}\left\{ \sum_{\boldsymbol{y} \in \mathbb{Y}} \left\vert p_{xy}-p_{x^{\prime }y}\right\vert \right\}.
\end{align*}

\begin{definition}
\label{arc_deletion}
Let the diameter $d^+_j$ of $\boldsymbol{Y}$ to $\boldsymbol{X}_j$  be defined by 
\begin{align*}
d^+_{j}= \frac{1}{2} \max_{\boldsymbol{x}_{\widehat{j}}\in \mathbb{X}_{\widehat{j}}}\max_{\boldsymbol{x}_j\boldsymbol{,x}_j^{\prime }\in \mathbb{X}_j}\left\{ \sum_{\boldsymbol{y} \in \mathbb{Y}}\left\vert p_{xy}-p_{x^{\prime }y}\right\vert \right\}.
\end{align*}
\end{definition}

So the diameter $d^+_j$ is the maximum extra effect varying the value of $\boldsymbol{x}_j$ can have on the distribution of $\boldsymbol{Y}$ for any fixed value $\boldsymbol{x}_{\widehat{j}}\in \mathbb{X}_{\widehat{j}}$ of the other variables.  Notice in particular that
\begin{align*}
\boldsymbol{Y}\amalg \boldsymbol{X}_j|\boldsymbol{X}_{\widehat{j}}\Leftrightarrow d^+_j=0.
\end{align*}
Thus in a formal sense, $d^+_j$ is a measure of the \textit{extent} by which this conditional independence is violated and the merit of knowing the value of $\boldsymbol{X}_j$ might have once we knew the value of $\boldsymbol{X}_{\widehat{j}}$.  We now have the simple but pleasing additive relationship between $d^{+}(P)$ and $d^+_j$, $j=1.2,\ldots ,k$:

\begin{theorem}
Under the notation above 
\begin{align*}
d^{+}(P)\leq \sum_{j=1}^k \delta _j .
\end{align*}
\end{theorem}

\begin{proof}
Here we simply use the triangle inequality to bound $d^{+}(P)$ changing the entries of the conditioning variables $\boldsymbol{X}_j$ one at a time.  So if $\boldsymbol{x}=\left( \boldsymbol{x}_1,\boldsymbol{x}_2,\ldots ,\boldsymbol{x}_k\right) ,\boldsymbol{x}^{\prime }=\left( \boldsymbol{x}_1^{\prime },\boldsymbol{x}_2^{\prime },\ldots ,\boldsymbol{x}_k^{\prime }\right) ,\boldsymbol{x}(0)=\boldsymbol{x},\boldsymbol{x}(1)=\left( \boldsymbol{x}_1^{\ast },\boldsymbol{x}_2,\ldots ,\boldsymbol{x}_k\right) ,\boldsymbol{x}(2)=\left( \boldsymbol{x}_1^{\prime },\boldsymbol{x}_2^{\ast },\mathbf{x}_3\ldots ,\boldsymbol{x}_k\right),\ldots ,\boldsymbol{x}(k)=\left( \boldsymbol{x}_1^{\prime },\boldsymbol{x}_2^{\prime },\ldots ,\boldsymbol{x}_k^{\ast }\right)$, then 
\begin{eqnarray*}
\max_{\boldsymbol{x,x}^{\prime }\in \mathbb{X}}\left\{ \sum_{j=1}^{n^{\prime }}\left\vert p_{xy}-p_{x^{\prime }y}\right\vert \right\} &\leq &\sum_{i=1}^{k}\max_{\boldsymbol{x}(i)\boldsymbol{,x}(i-1)\in \mathbb{X}}\left\{ \sum_{j=1}^{n^{\prime }}\left\vert p_{\boldsymbol{x}(i)y}-p_{\boldsymbol{x}(i-1)y}\right\vert \right\} \\
&=&\sum_{i=1}^{k}\max_{\boldsymbol{x}_{\widehat{j}}\in \mathbb{X}_{\widehat{j}}}\max_{\boldsymbol{x}_j\boldsymbol{,x}_j^{\ast }\in \mathbb{X}_j}\left\{ \sum_{j=1}^{n^{\prime }}\left\vert p_{xy}-p_{x^{^{\ast }}y}\right\vert \right\}.
\end{eqnarray*}
since by definition of $\boldsymbol{x}$ and $\boldsymbol{x}^\prime$, the adjacent CPTs appearing in the sum above differ only in the $i^{th}$ entry.
\end{proof} 

\subsection{Robustness to approximation by a sparser BN and edge deletion}
Often a BN is chosen to be sparser than it would be were we to have more information or time.  This happens for a variety of reasons.  For example when eliciting a BN we often ask for the list of the \emph{most important} variables on which a specific variable $X$ might depend, defining this phrase by asking that variables not included in the list could be expected to have only a small influence on $X$.  This restriction is imposed because it is difficult for a client to think clearly about the interrelationships between more than a handful of variables. Increasing the number of different joint levels on the conditioning variables quadratically increases the number of entries in the CPT that need to be elicited. 

If data is used to inform the model choice then a severe penalty is often imposed or is implicit to ensure the selection of models with smaller size parameter spaces -- which in this context usually implies sparser associated graphs.  Finally, for reasons of implementability, it is quite common for a search of candidate BNs to include only those graphs whose nodes have no more than a fixed number of parents: the limit often set to be two or three, see e.g.  \cite{cussens2011}. 

It is therefore very important to properly understand the implications of these potential over-simplifications on the robustness of the BN model. We present some corollaries on the use of these results.

\subsection{Edge Deletion}
When constructing BN systems we seek the model which best describes the underlying physical process.  However, in reality we are often limited by restrictions on resources and time which lead us to instead develop a model that is as large as it need be.  The size of a BN can grow exponentially by adding more variables, relationships or even states.  So eliciting these probabilities becomes problematic.  Therefore, a frequent simplification of a model can be to delete unnecessary edges that impact only a little on hte outputs of the system. 

The methodology we have introduced in earlier sections can be directly applied to this problem in order to quantify the cost of removing a certain variable from the parent set of another variable. 

\begin{example}
Return once again to the CPT of `Tree Condition' given its binary parent `Drought' and three-state parent `Annual Rainfall', as in Example \ref{example_TreeCondition_CPT}.  Denote the rows of $P$ by $\boldsymbol{p}_i$ for $i=1,2,\hdots,6$.  If we are interested in the effects of deleting the arc Drought $\longrightarrow$ Tree Condition, then we can use Definition \ref{arc_deletion}  directly: 
\begin{align*}
\max_{ \{yes,no \} } \max_{ \{Below, Avg., Above\} } \{ d_V(\boldsymbol{p}_i, \boldsymbol{p}_{i^\prime}) \} &= 
%\max \{d_V(\boldsymbol{p}_1,\boldsymbol{p}_4), d_V(\boldsymbol{p}_2,\boldsymbol{p}_5),d_V(\boldsymbol{p}_3,\boldsymbol{p}_6) \} \\
 0.6
\end{align*}
Alternatively, deleting arc Annual Rainfall $\longrightarrow$ Tree Condition gives us
\begin{align*}
\max_{\{ Below, Avg., Above\} } \max_{\{ yes,no \}} \{ d_V(\boldsymbol{p}_i, \boldsymbol{p}_{i^\prime}) \} &= 
%&= \{ \max\{ d_V(\boldsymbol{p}_1,\boldsymbol{p}_2),d_V(\boldsymbol{p}_2,\boldsymbol{p}_3),d_V(\boldsymbol{p}_1,\boldsymbol{p}_3)\}, \max\{d_V(\boldsymbol{p}_4,\boldsymbol{p}_5), d_V(\boldsymbol{p}_5,\boldsymbol{p}_6), d_V(\boldsymbol{p}_4,\boldsymbol{p}_6)\} \} \\
 0.2
\end{align*}
Below we present the full BN with each edge annotated with the quantitative effect on the child if we delete the parent.
\begin{figure}[H]
\begin{center}
\includegraphics[width=\textwidth]{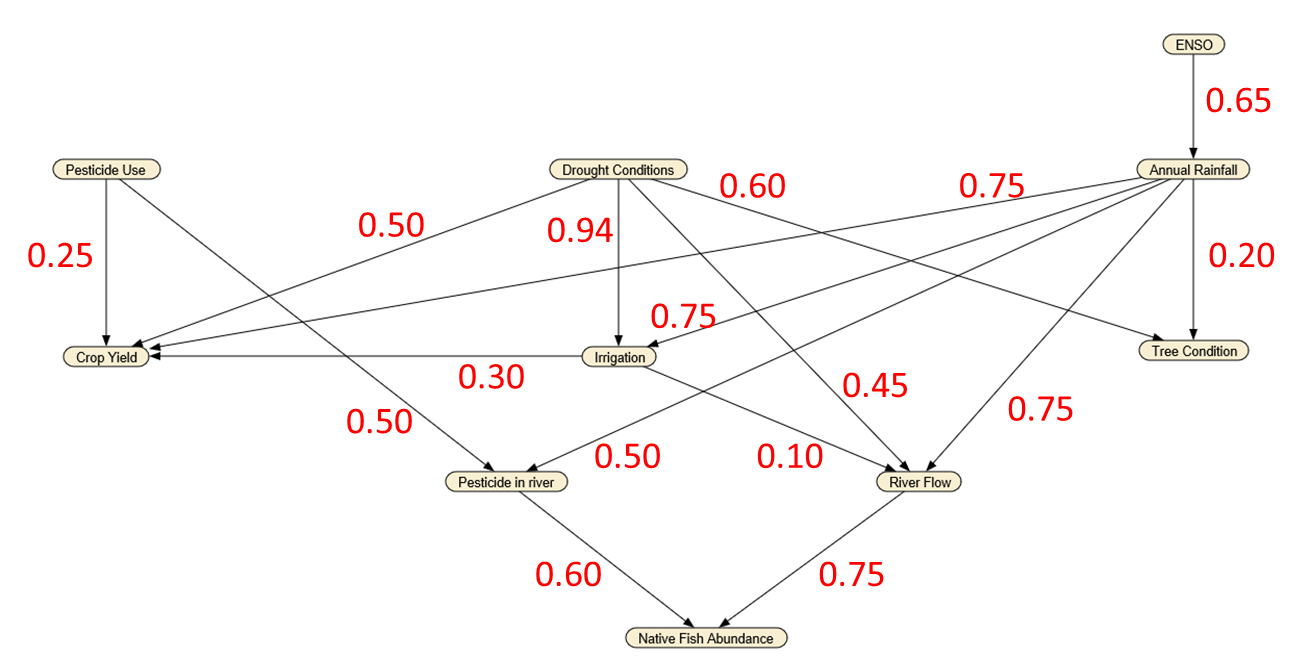}
\caption{Effects of Edge Deletion in the Native Fish BN.}
\end{center}
\end{figure}
\end{example}

The higher the value of $\delta_j$ the larger the effect of deleting the arc $\boldsymbol{X}_j \longrightarrow \boldsymbol{Y}$, because it is indicating that the corresponding rows of other parent responses are dissimilar. 

Now of course this model is already well designed and so we find that most edges need to be there. However, the possible exception if the edge from Irrigation to River Flow.  In candidate models devised early in an analysis, or in large BNs, we often find that many edges can be omitted without strong effects on the outputs of the system.

Note that this process scales up easily to handle parent sets of discrete BNs with numerous nodes and large numbers of states, due to the variation distance being a simple and transparent piece of arithmetic. 

\subsection{Level amalgamation}
One practical issue found by discrete BN modellers is the number of levels each random variable within the system should be assigned. Obviously there is a trade-off here.  The finer the division of levels, the more nuanced the BN can be.  On the other hand, the fewer the number of levels, the easier it will be to faithfully elicit or efficiently estimate the probabilities within a BN.  One advantage of using the variation approach for robustness is that such considerations can be taken under the same technical umbrella as other necessary approximations, such as edge deletion.  We simply evaluate the potential that such simplifications might have on the distribution of the attributes of the problem, just as when considering whether or not to keep a weak edge in the system. 

When considering amalgamating levels within a specified random variable we must ensure that the interpretation of the states can still be understood and quantified by experts.  When those variables are ordinal and have a monotonic relationship with its neighbours we would therefore recommend that you only consider amalgamating consecutive node levels.  For example, the node of Annual Rainfall in Example \ref{example_TreeCondition_CPT} had levels $\{Below Average, Average, Above Average\}$ and if we were to combine Below Average with Above Average we would not have a cohesive state for the experts to quantify.

The second step in level amalgamation is deciding how to combine the probabilities associated with those levels that are to be amalgamated.  We recommend taking a simple row average, because the convexity of variation distance tells us that this will enable us to avoid increasing the diameter of the original CPT and more importantly this method does not require additional information (which would otherwise be needed to perform a weighted row average). 

Occasionally the modeller or expert will have an intuitive feel for which states should be combined, possibly from past experience or relevant data.  However, sometimes this may not be obvious.  In the latter case, we can calculate the variation distance between the considered states and combine the closest states first, then find the next closest state and add to the amalgamation iteratively until the combination appears to induce significant variation distance from the original.

\begin{example}
For example, let us once again consider the `Tree Condition' CPT from \ref{example_TreeCondition_CPT}.  To reduce the three-state node Annual Rainfall to a binary node we could either combine Below Average with Average (case i) or Above Average with Average (case ii).  To decide we compare the variation distance between rows $\boldsymbol{p}_i$:
\begin{align*}
\text{Case i: \hspace{1cm}} &\max\{d_V(\boldsymbol{p}_1,\boldsymbol{p}_2), d_V(\boldsymbol{p}_4, \boldsymbol{p}_5)\} = 0.1 \\
\text{Case ii: \hspace{1cm}} &\max\{d_V(\boldsymbol{p}_2,\boldsymbol{p}_3), d_V(\boldsymbol{p}_5, \boldsymbol{p}_6)\} = 0.1
\end{align*}
In this instance we can arbitrarily chose between the two, so we shall opt for case (i) to form the amalgamated state `Average or Below' using a simple average of relevant rows to obtain $P^\prime$:
\begin{equation*}
P = \begin{pmatrix}
0.20 & 0.60 & 0.20 \\
0.25 & 0.60 & 0.15 \\
0.30 & 0.60 & 0.10 \\
0.70 & 0.25 & 0.05 \\
0.80 & 0.18 & 0.02 \\
0.90 & 0.09 & 0.01
\end{pmatrix}, \hspace{1cm}
P^\prime = \begin{pmatrix}
0.225 & 0.60 & 0.175 \\
0.30 & 0.60 & 0.10 \\
0.725 & 0.215 & 0.06 \\
0.90 & 0.09 & 0.01
\end{pmatrix}.
\end{equation*}
Calculating that $d_V^+(P,P^\prime)=0.075$ in results such as Theorem \ref{theorem_totvar_diameter} shows that the effect of using this amalgamated CPT rather than the original is small. 
\end{example} 

\section{Some principled strategies for BN creation}
\label{strategies_section}
Obviously the evaluation criteria we indicate here can be embedded into a formal protocol.  However, there are many considerations that a user has to consider before undertaking model construction: transparency of the model, computational issues, elicitation constraints and so forth, which vary in importance depending on the context of the model building.  So setting a bound on any effects or perturbations against differing approaches is often best undertaken informally.  However, we acknowledge that the framework we have presented here is sufficiently formal to admit generalisation and this is work that we plan to undertake next. 

To implement our techniques as efficiently as possible we recommend two differing approaches tailored to the specific circumstances of the modeller.  Firstly, there are occasions in which we have obtained provisional information from one expert who can recommend nodes, levels, interactions and provisional CPTs before undertaking a more formal elicitation conference with multiple experts.  Such was the case in the pollinator example discussed in \cite{Pollinator_workshop_JAR_paper} and \cite{bookchapter}.  In this particular scenario we can begin to design the analysis by using the bounds discussed earlier on the preliminary values stated by the expert.  We recommend starting by eliciting attributes and nodes of interest before working systematically backwards along the chain of inference to discover parent nodes and conditional independences, performing variation measures on preliminary CPT values to determine the efficacy of including variables in the model.  Of course after the full elicitation has taken place the robustness analyses suggested above can be repeated for a final sensitivity analysis. 

In situations when we are starting the model with no such preliminary information it may be wisest to attempt to elicit the value of the diameter of the CPT directly, before eliciting the full matrix, so that full elicitation is not undertaken before we can derive concrete bounds on the usefulness of this data harvesting exercise.  This can then be bounded and decisions undertaken on whether to include certain variables in the chain or not.  To elicit the diameter directly we need to ascertain the largest differences between rows of a CPT, which corresponds to requesting the ``best case scenario'' probabilities and the ``worst case scenario'' probabilities before calculating the variation distance between the two. 

By following this procedure we therefore continuously appraise and compare each possible simplification  against the potential accuracy of an analysis, weighted against the issues provided by a simpler representation of a model.  We have demonstrated above that it can be \textit{proved} that in many cases the effects of various simplifications are often very small, and approximations based on these simplifications are justified from a pragmatic point of view.  We also note that some of the best approximations to use are often not the ones currently undertaken in practice.  For example we often find that using an approximation which deletes an edge can cause significant changes, whilst allowing dependence only on subsets of levels performs much better. 

\section{Conclusions}
We have demonstrated here how the properties of variation distance can be harnessed to study the robustness of a discrete BN, if certain target variables are known a priori to be those of primary interest.  Although all our illustrative examples in this paper have perforce been of moderate size, our methods become ever more useful as the number of nodes in the BN increases. Even when this number is huge, we can show it is possible to identify a priori which features of the full joint distribution will have the strongest impact on the target variables of interest, and therefore employ effective and expedient approximations to make inferences which are both provably accurate and feasible for the task at hand.  In such models, since the simple paths between learned variables and attribute variables is typically much longer, it is possible to formally demonstrate that some remote variables are just not worth eliciting directly, but should be marginalised over.  

The approach we have introduced in this paper relies on the well-studied variation distance which naturally embeds conditional independence relationships between variables.  We have therefore devised a seamless way of looking at perturbed versions of a BN in a manner which enables us to apply the same devices to generic effects, be these perturbations associated with edge deletion, the knock-on effect of learning certain variables, misquoting probabilities within CPTs or changing the number of levels for nodes.  All of these different alterations can be compared on an equal footing whereas previous work usually depended on model selection of BNs, using methodology such as Bayes Factors, which focused on one particular perturbation at a time.  We note that many of the techniques to communicate visually the bounding effect of one node on a target node can be straightforwardly adapted and reapplied to this domain, such as the use of heatmaps demonstrated in  \cite{Albrecht2014}.  

There is of course much work to still be undertaken in this field, starting with refining the bounds we have developed here.  The robustness studies we have introduced can also be applied to context-specific BNs where we have a natural trade-off between the number of probabilities to elicit and the robustness of the model.  Typically if we elicit less probabilities for a context-specific model we can show that we weaken the robustness of the system due to the constraint of forcing inputs to be the same.  Similarly within this paper we have had no space to consider the robustness to the choice of probability distribution on the entries of the CPTs of a BN.  In \cite{smith-daneshkhah2010}, BN robustness associated with the inputs of the distribution in terms of the local DeRobertis distance is studied.  \cite{smith-rigat2012} provide bounds on posterior variation distances.  Therefore a fairly straightforward extension to the variation bounds we have presented here can be developed by carefully combining our results with the local DeRobertis distance to provide a comprehensive robustness analysis when necessary.  Essentially we can show that with sufficient data and global independence assumptions the most robust CPTs are the ones whose probabilities are best known. 

Finally, recent theoretical advances have suggested that if a Bayesian accepts the their model is only approximate, the M-open scenario \citep{Bernardo_Smith_1996}, then Bayesian learning using Bayes Rule may not be optimal and that other updating rules based on divergences other than the KL divergence should be considered.  There are exciting new possibilities of combining this technology with the robustness methods described here when that divergence is defined as the variation distance \citep{Jack_Smith_Holmes}.

Our ideas also apply directly to the Dynamic BN where the robustness of the system can be even more important because the dynamic nature of the problem makes the model much more complex.  Throughout this paper, for simplicity, we have considered only robustness as it applies to finite discrete BNs.  However the whole technology we describe here of course translates seamlessly into tools for examining continuous and mixed Bayesian Networks.  Using the variation distance on these highly structured and complex Markov Processes using the approach highlighted here can help us to determine the robustness of DBNs to dynamic effects.  Work in this more general setting has already begun.  We hope we have demonstrated that this is actually a very fruitful way of addressing robustness within this family of graphical models. 

\newpage
\bibliographystyle{plainnat}
\bibliography{D:/skw/Documents/University_of_Warwick/Thesis_Library/master}{}

\begin{thebibliography}{35}
\providecommand{\natexlab}[1]{#1}
\providecommand{\url}[1]{\texttt{#1}}
\expandafter\ifx\csname urlstyle\endcsname\relax
  \providecommand{\doi}[1]{doi: #1}\else
  \providecommand{\doi}{doi: \begingroup \urlstyle{rm}\Url}\fi

\bibitem[Albrecht et~al.(2014{\natexlab{a}})Albrecht, Nicholson, and
  Whittle]{albrecht_nicholson_whittle}
D.~Albrecht, A.E. Nicholson, and C.~Whittle.
\newblock Structural sensitivity for the knowledge engineering of bayesian
  networks.
\newblock In \emph{Probabilistic Graphical Models: 7th European Workshop},
  chapter~1. Springer, 2014{\natexlab{a}}.

\bibitem[Albrecht et~al.(2014{\natexlab{b}})Albrecht, Nicholson, and
  Whittle]{Albrecht2014}
D.~W. Albrecht, A.~E. Nicholson, and C.~Whittle.
\newblock Structural sensitivity for the knowledge engineering of bayesian
  networks.
\newblock In L.~C. van~der Gaag and A.~J. Feelders, editors, \emph{Proceedings
  of the Seventh European Workshop on Probabilistic Graphical Models (PGM
  2014)}, volume 8754, pages 1--16. Springer, Cham, 2014{\natexlab{b}}.

\bibitem[Barons et~al.(2018{\natexlab{a}})Barons, Hanea, Wright, Baldock,
  Wilfert, Chandler, Datta, Fannon, Hartfield, Lucas, Ollerton, Potts, and
  Carreck]{Pollinator_workshop_JAR_paper}
M.~J. Barons, A.~M. Hanea, S.~K. Wright, K.~C.~R. Baldock, L.~Wilfert,
  D.~Chandler, S.~Datta, J.~Fannon, C.~Hartfield, A.~Lucas, J.~Ollerton, S.~G.
  Potts, and N.~L. Carreck.
\newblock Assessment of the response of pollinator abundance to environmental
  pressures using structured expert elicitation.
\newblock \emph{Journal of Apicultural Research}, 57\penalty0 (5):\penalty0
  593--604, 2018{\natexlab{a}}.

\bibitem[Barons et~al.(2018{\natexlab{b}})Barons, Wright, and
  Smith]{bookchapter}
M.~J. Barons, S.~K. Wright, and J.~Q. Smith.
\newblock Eliciting probabilistic judgements for integrating decision support
  systems.
\newblock In Dias~L. C., A.~Morton, and J.~Quigley, editors, \emph{Elicitation:
  The science and art of structuring judgement}, chapter~17, pages 445--478.
  Springer, 2018{\natexlab{b}}.

\bibitem[Bernardo and Smith(1994)]{Bernardo_Smith_1996}
J.~M. Bernardo and A.~F.~M. Smith.
\newblock \emph{Bayesian Theory}.
\newblock Wiley, Chichester, 1994.

\bibitem[Boneh(2010)]{boneh_2010}
T.~Boneh.
\newblock \emph{Ontology and Bayesian decision networks for supporting the
  meteorological forecasting process.}
\newblock PhD thesis, Clayton School for Information Technology, Monash
  University, 2010.

\bibitem[Chan and Darwiche(2005)]{chan_darwiche_2005}
H.~Chan and A.~Darwiche.
\newblock A distance measure for bounding probabilistic belief change.
\newblock \emph{International Journal of Approximate Reasoning}, 38\penalty0
  (2):\penalty0 149--174, 2005.

\bibitem[Coup\'{e} and van~der Gaag(2002)]{coupe_gaag_2002}
V.~M. Coup\'{e} and L.~C. van~der Gaag.
\newblock Properties of sensitivity analysis of bayesian belief networks.
\newblock \emph{Annals of Mathematics and Artificial Intelligence}, 36\penalty0
  (4):\penalty0 323--356, 2002.

\bibitem[Cowell et~al.(1999)Cowell, Dawid, Lauritzen, and
  Spiegelhalter]{cowell_dawid_lauritzen_spiegelhalter}
R.~G. Cowell, A.~P. Dawid, S.~L. Lauritzen, and D.~J. Spiegelhalter.
\newblock \emph{Probabilistic Networks and Expert Systems.}
\newblock Springer-Verlag, New York, 1999.

\bibitem[Cowell et~al.(2007)Cowell, Verrall, and Yoon]{cowell_et_al_2007}
R.~G. Cowell, R.~J. Verrall, and Y.~K. Yoon.
\newblock Modeling operational risk with bayesian networks.
\newblock \emph{Journal of Risk and Insurance}, 74\penalty0 (4):\penalty0
  795--827, 2007.

\bibitem[Cussens(2011)]{cussens2011}
J.~Cussens.
\newblock Bayesian network learning with cutting planes.
\newblock In F.~Cozman and A.~Pfeffer, editors, \emph{Proceedings of the 27th
  Conference on Uncertainty in Artificial Intelligence}, pages 153--160, 2011.

\bibitem[EFSA(2014)]{EFSA}
EFSA.
\newblock Guidance on expert knowledge elicitation in food and feed safety risk
  assessment.
\newblock \emph{EFSA Journal}, 12\penalty0 (6):\penalty0 3734, 2014.

\bibitem[Friedman et~al.(1997)Friedman, Geiger, and Goldszmidt]{friedman}
N.~Friedman, D.~Geiger, and M.~Goldszmidt.
\newblock Bayesian network classifiers.
\newblock \emph{Machine Learning}, 29\penalty0 (2-3):\penalty0 131--163, 1997.

\bibitem[G\'{o}mez-Villegas et~al.(2013)G\'{o}mez-Villegas, A., Main, and
  Susi]{BN_KL}
G\'{o}mez-Villegas, M.~A., P.~Main, and R.~Susi.
\newblock The effect of block parameter perturbations in gaussian bayesian
  networks: sensitivity and robustness.
\newblock \emph{Information Sciences}, 222:\penalty0 439--458, 2013.

\bibitem[Gustafson and Wasserman(1995)]{Gustafson&Wasserman}
P.~Gustafson and L.~Wasserman.
\newblock Local sensitivity diagnostics for bayesian inference.
\newblock \emph{The Annals of Statistics}, 23\penalty0 (6):\penalty0
  2153--2167, 1995.

\bibitem[Jewson et~al.(2018)Jewson, Smith, and Holmes]{Jack_Smith_Holmes}
J.~Jewson, J.~Q. Smith, and C.~Holmes.
\newblock Principles of bayesian inference using general divergence criteria.
\newblock \emph{Entropy}, 20\penalty0 (6):\penalty0 442, 2018.

\bibitem[Korb and Nicholson(2010)]{korb-nicholson2010}
K.~B. Korb and A.~E. Nicholson.
\newblock \emph{Bayesian artificial intelligence, second edition}.
\newblock CRC press, 2010.

\bibitem[Laskey(1995)]{laskey-1995}
K.~B. Laskey.
\newblock Sensitivity analysis for probability assessments in bayesian
  networks.
\newblock In \emph{IEEE Transactions on Systems, Man, and Cybernetics},
  volume~25, pages 901--909, 1995.

\bibitem[Laskey and Mahoney(2000)]{laskey_mahoney_2000}
K.~B. Laskey and S.~M. Mahoney.
\newblock Network engineering for agile belief network models.
\newblock \emph{IEEE Transactions on Knowledge and Data Engineering},
  12\penalty0 (4):\penalty0 487--498, 2000.

\bibitem[Lauritzen(1996)]{lauritzen1996}
S.~L. Lauritzen.
\newblock \emph{Graphical models}.
\newblock Clarendon Press, Oxford, 1996.

\bibitem[Lauritzen and Spiegelhalter(1988)]{lauritzen_spiegelhalter_1988}
S.~L. Lauritzen and D.~J. Spiegelhalter.
\newblock Local computations with probabilities on graphical structures and
  their application to expert systems.
\newblock \emph{Journal of the Royal Statistical Society. Series B
  (Methodological)}, pages 157--224, 1988.

\bibitem[Leonelli et~al.(2017)Leonelli, G\"{o}rgen, and
  Smith]{leonelli_gorgen_smith}
M.~Leonelli, C.~G\"{o}rgen, and J.~Q. Smith.
\newblock Sensitivity analysis in multilinear probabilistic models.
\newblock \emph{Information Sciences}, 411:\penalty0 84--97, 2017.

\bibitem[Nicholson et~al.(2010)Nicholson, Woodberry, and Twardy]{NativeFish}
A.~Nicholson, O.~Woodberry, and C.~Twardy.
\newblock The ``native fish'' bayesian networks.
\newblock Technical report, Technical Report 2010/3, Bayesian Intelligence.,
  2010.

\bibitem[Nicholson and Flores(2011)]{FloresNicholson}
A.~E. Nicholson and M.~J. Flores.
\newblock Combining state and transition models with dynamic bayesian networks.
\newblock \emph{Ecological Modelling}, 222\penalty0 (3):\penalty0 555--566,
  2011.

\bibitem[Nicholson and Jitnah(1998)]{Nicholson_mutual_information}
A.~E. Nicholson and N.~Jitnah.
\newblock Using mutual information to determine relevance in bayesian networks.
\newblock In \emph{Pacific Rim International Conference on Artificial
  Intelligence}, pages 399--410. Springer, Berlin, Heidelberg, 1998.

\bibitem[O'Hagan et~al.(2006)O'Hagan, Buck, Daneshkhah, Eiser, Garthwaite,
  Jenkinson, Oakley, and Rakow]{elicitation_book}
A.~O'Hagan, C.~E. Buck, A.~Daneshkhah, J.~R. Eiser, P.~H. Garthwaite, D.~J.
  Jenkinson, J.~E. Oakley, and T.~Rakow.
\newblock \emph{Uncertain judgements: eliciting experts' probabilities.}
\newblock John Wiley \& Sons, 2006.

\bibitem[O'Neill(2009)]{o'neill-2009}
B.~O'Neill.
\newblock Importance sampling for bayesian sensitivity analysis.
\newblock \emph{International Journal of Approximate Reasoning}, 50\penalty0
  (2):\penalty0 270--278, 2009.

\bibitem[Renooij(2010)]{renooij-2010}
S.~Renooij.
\newblock Bayesian network sensitivity to arc-removal.
\newblock In \emph{Proceedings of the Fifth European Workshop on Probabilistic
  Graphical Models, Helsinki.}, page 233, 2010.

\bibitem[Scutari and Denis(2014)]{scutari_book}
M.~Scutari and J.~B. Denis.
\newblock \emph{Bayesian networks: with examples in R.}
\newblock CRC press, 2014.

\bibitem[Smith(2010)]{jim_book}
J.~Q. Smith.
\newblock \emph{Bayesian decision analysis: principles and practice.}
\newblock Cambridge University Press, 2010.

\bibitem[Smith and Daneshkhah(2010)]{smith-daneshkhah2010}
J.~Q. Smith and A.~Daneshkhah.
\newblock On the robustness of bayesian networks to learning from non-conjugate
  sampling.
\newblock \emph{International Journal of Approximate Reasoning}, 51\penalty0
  (5):\penalty0 558--572, 2010.

\bibitem[Smith and Rigat(2012)]{smith-rigat2012}
J.~Q. Smith and F.~Rigat.
\newblock Isoseparation and robustness in finite parameter bayesian inference.
\newblock \emph{Annals of the Institute of Statistical Mathematics},
  64\penalty0 (3):\penalty0 495--519, 2012.

\bibitem[Smith et~al.(2015)Smith, Barons, and Leonelli]{SBL2015}
J.~Q. Smith, M.~J. Barons, and M.~Leonelli.
\newblock Coherent frameworks for statistical inference serving integrating
  decision support systems.
\newblock \emph{arXiv preprint arXiv:1507.07394}, 2015.

\bibitem[Wright(2018)]{skw_thesis}
S.~K. Wright.
\newblock \emph{Robustness in Bayesian Networks.}
\newblock PhD thesis, Department of Statistics, University of Warwick, 2018.

\bibitem[Zaragoza et~al.(2011)Zaragoza, Sucar, and Morales]{Zaragoza}
J.~C. Zaragoza, L.~E. Sucar, and E.~F. Morales.
\newblock A two-step method to learn multidimensional bayesian network
  classifiers based on mutual information measures.
\newblock In \emph{Proceedings of the Twenty-Fourth International Florida
  Artificial Intelligence Research Society Conference (FLAIRS)}, pages
  644--649, 2011.

\end{thebibliography}

\newpage
\begin{appendices}
\section{Proofs}
\subsection{Proof of Lemma \ref{mixing_lemma}}
\label{appendixA1}
To prove Lemma \ref{mixing_lemma} we shall first need some simple intermediary results: \\ \\
Let $\mathbf{p}, \mathbf{q}_0, \mathbf{q}_1$ be three vectors of mass functions and define 
\begin{equation*}
\mathbf{q}_\alpha \triangleq (1-\alpha)\mathbf{q}_0 + \alpha \mathbf{q}_1,
\end{equation*}
for $0 \leq \alpha \leq 1$. 
\begin{lemma}
\label{Lemma A1}
\begin{equation*}
d_V(\mathbf{p},\mathbf{q}_\alpha) \leq (1-\alpha)d_V(\mathbf{p},\mathbf{q}_0)+\alpha d_V(\mathbf{p},\mathbf{q}_1)
\end{equation*}
\end{lemma}
\begin{proof}
Note that for $0\leq \alpha \leq 1$,
\begin{equation*} 
\mathbf{p}=(1-\alpha )\mathbf{p}+\alpha \mathbf{p}
\end{equation*}
so that 
\begin{align*}
2d_{V}(\mathbf{p},\mathbf{q}_{\alpha }) &=\sum_{j=1}^n\left\vert p_{j}-(1-\alpha )q_{0j}-\alpha q_{1j}\right\vert  \\
&=\sum_{j=1}^n\left\vert (1-\alpha)p_j+\alpha p_j-(1-\alpha )q_{0j}-\alpha q_{1j}\right\vert \\
&=\sum_{j=1}^n\left\vert (1-\alpha )(p_{j}-q_{0j})+\alpha
(p_{j}-q_{1j})\right\vert 
\end{align*}
Next note that if $(p_{j}-q_{0j}),(p_{j}-q_{1j})$ are the same sign then 
\begin{equation*}
\left\vert (1-\alpha )(p_{j}-q_{0j})+\alpha (p_{j}-q_{1j})\right\vert =(1-\alpha )\left\vert p_{j}-q_{0j}\right\vert +\alpha \left\vert p_{j}-q_{1j}\right\vert 
\end{equation*}
whilst if $(p_{j}-q_{0j}),(p_{j}-q_{1j})$ are of different sign then 
\begin{align*}
\left\vert (1-\alpha )(p_{j}-q_{0j})+\alpha (p_{j}-q_{1j})\right\vert  &\leq \max \left\{ (1-\alpha )\left\vert p_{j}-q_{0j}\right\vert ,\alpha \left\vert p_{j}-q_{1j}\right\vert \right\}  \\
&\leq (1-\alpha )\left\vert p_{j}-q_{0j}\right\vert +\alpha \left\vert p_{j}-q_{1j}\right\vert .
\end{align*}
It follows that 
\begin{align*}
2d_{V}(\mathbf{p},\mathbf{q}_{\alpha }) &\leq \sum_{j=1}^{n}(1-\alpha )\left\vert p_{j}-q_{0j}\right\vert +\alpha \left\vert p_{j}-q_{1j}\right\vert  \\
&=(1-\alpha )\sum_{j=1}^{n}\left\vert p_{j}-q_{0j}\right\vert+\alpha \sum_{j=1}^{n} \left\vert p_{j}-q_{1j}\right\vert  \\
&=2\left\{ (1-\alpha )d_{V}(\mathbf{p},\mathbf{q}_{0})+\alpha d_{V}(\mathbf{p},\mathbf{q}_{1})\right\} 
\end{align*}
proving the result.
\end{proof}

This results proves that the variation distance between a mass function $\mathbf{p}$ and the mixture of two others is less than the mixture of the mass function $\mathbf{p}$ and the two extremal distributions $\mathbf{q}_0, \mathbf{q}_1$. We can now extend this result so that it applies to any finite mixtures: let $\pi = (\pi_1,\pi_2,\ldots,\pi_n)$ with $\sum_{i=1}^n \pi_i =1$, $\mathbf{q}_i=(q_{i1},q_{i2},\ldots,q_{in})$ and define
\begin{equation*}
\mathbf{q}_\pi \triangleq \sum_{i=1}^n \pi_i \mathbf{q}_i.
\end{equation*}
\begin{lemma}
\label{Lemma A2}
\begin{equation*}
d_V(\mathbf{p},\mathbf{q}_\pi) \leq \sum_{i=1}^n \pi_i d_V(\mathbf{p},\mathbf{q}_i).
\end{equation*}
\end{lemma}
\begin{proof}
We shall proceed using induction. For the case $n=2$ we have $\mathbf{q}_\pi =\pi_1 q_1 + \pi_2 q_2$ with $\pi_1 +\pi_2 =1$ which means $\mathbf{q}_\pi=(1-\pi_2)q_1+\pi_2 q_2$. By applying the previous lemma we know that 
\begin{align*}
d_V(p,\mathbf{q}_\pi) &\leq (1-\pi_2)d_V(p,q_1) + \pi_2 d_V(p,q_2) \\
&= \pi_1 d_V(p,q_1) + \pi_2 d_V(p,q_2) \\
&= \sum_{i=1}^2 \pi_i d_V(p,q_i).
\end{align*} 
So this statement certainly holds true for $n=2$. Assume that it is true for $n-1$. Then note that for the case $n=n$:
\begin{align*}
\mathbf{q}_{\mathbf{\pi }}=\sum_{i=1}^{n}\pi _{i}\mathbf{q}_{i} &= \pi_1 \mathbf{q}_1+\pi_2 \mathbf{q}_2 + \ldots + \pi_{n-1}\mathbf{q}_{n-1}+\pi_n \mathbf{q}_n \\
&= (\pi_1 + \pi_2 + \ldots + \pi_{n-1})\mathbf{q}_0 + \pi_n \mathbf{q}_n \\
&= (1-\pi_n) \mathbf{q}_0 + \pi_n \mathbf{q}_n \\
&=(1-\alpha )\mathbf{q}_{0}+\alpha \mathbf{q}_{n} 
\end{align*}
where $\alpha \triangleq \pi_{n}$ and for $i=1,2,\ldots n-1$, $\pi _{i}^{\prime }\triangleq \pi _{i}(1-\pi_{n})^{-1}$, giving
\begin{equation*}
\mathbf{q}_{0}\triangleq \sum_{i=1}^{n-1}\pi _{i}^{\prime}\mathbf{q}_{i} = \frac{1}{1-\pi_n}\sum_{i=1}^{n-1} \pi_i q_i
\end{equation*}
By the last lemma
\begin{equation*}
d_{V}(\mathbf{p},\mathbf{q}_{\mathbf{\pi }})\leq (1-\alpha )d_{V}(\mathbf{p},\mathbf{q}_{0})+\alpha d_{V}(\mathbf{p},\mathbf{q}_{n}) 
\end{equation*}
where by the inductive hypothesis
\begin{equation*}
d_{V}(\mathbf{p},\mathbf{q}_{0})\leq \sum_{i=1}^{n-1}\pi_{i}^{\prime }d_{V}(\mathbf{p},\mathbf{q}_{i}).
\end{equation*}
Substituting the inductive hypothesis into our equation gives
\begin{align*}
d_{V}(\mathbf{p},\mathbf{q}_{\mathbf{\pi }}) &\leq
\sum_{i=1}^{n-1}(1-\alpha )\pi _{i}^{\prime }d_{V}(\mathbf{p},
\mathbf{q}_{i})+\alpha d_{V}(\mathbf{p},\mathbf{q}_{n}) \\
&=\sum_{i=1}^{n}\pi _{i}d_{V}(\mathbf{p},\mathbf{q}_{i})
\end{align*}
as required.
\end{proof}

\noindent \textbf{Lemma \ref{mixing_lemma}} stated that
\begin{equation*}
d_V(\mathbf{p}_{\pi^\prime},\mathbf{q}_\pi) \leq \sum_{i=1}^n \sum_{i^\prime=1}^{n^\prime} \pi_i \pi_{i^\prime}^\prime d_V(\mathbf{p}_{i^\prime},\mathbf{q}_i)
\end{equation*}
\begin{proof}
Here we simply use the symmetry of variation distance. From Lemma \ref{Lemma A2} we have that
\begin{equation}
\label{p1}
d_{V}(\mathbf{p}_{\mathbf{\pi }^{\prime }},\mathbf{q}_{\mathbf{\pi }})\leq \sum_{i=1}^{n}\pi _{i}d_{V}(\mathbf{p}_{\mathbf{\pi }^{\prime }},\mathbf{q}_{i}) 
\end{equation}
but for $i=1,2,\ldots ,n$ 
\begin{equation}
\label{p2}
d_{V}(\mathbf{p}_{\mathbf{\pi }^{\prime }},\mathbf{q}_{i})=d_{V}(\mathbf{q}_{i},\mathbf{p}_{\mathbf{\pi }^{\prime }})\leq
\sum_{i^{\prime }=1}^{n^{\prime }}\pi _{i}^{\prime }d_{V}(\mathbf{q}_{i},\mathbf{p}_{i^{\prime }})=\sum_{i^{\prime }=1}^{n^{\prime }}\pi_{i}^{\prime }d_{V}(\mathbf{p}_{i^{\prime }},\mathbf{q}_{i}) 
\end{equation}
So substituting Inequality \ref{p2} into Inequality \ref{p1} gives us our result.
\end{proof}

\subsection{Proof of Theorem \ref{theorem_totvar_diameter}}
\label{appendixA2}
\noindent Theorem \ref{theorem_totvar_diameter} stated that if we let $\mathbf{\pi}_{1}$ and $\mathbf{\pi}_{2}$ be two possible margins of $X$, suppose that $P$ is the (shared) CPT of $Y|X$ and that $\mathbf{\rho }_{1}$ and $\mathbf{\rho}_{2}$ are the margins of $Y$. Then 
\begin{equation*}
d_{V}(\mathbf{\rho }_{1},\mathbf{\rho }_{2})\leq d^+(P)d_{V}(\mathbf{\pi }_{1},\mathbf{\pi }_{2}) 
\end{equation*}
\begin{proof}
For $i=1,2,\dots,n$ let $\pi_1 \wedge \pi_2 \triangleq \min_i(\pi_1(i),\pi_2(i))$ and let us define
\begin{align*}
\pi_{1 \wedge 2} = \frac{\pi_1 \wedge \pi_2}{\sum (\pi_1 \wedge \pi_2)} = \frac{\pi_1 \wedge \pi_2 }{\beta}
\end{align*}
So by definition $1-\beta = d_V(\pi_1,\pi_2)$. \\
\indent Let 
\begin{align*}
\pi_1^\prime &= d_V(\pi_1, (\pi_1 \wedge \pi_2)) \\
\pi_2^\prime &= d_V(\pi_2, (\pi_1 \wedge \pi_2))
\end{align*}
However, we would like $\pi_1^\prime$ and $\pi_2^\prime$ to be densities, so let:
\begin{align*}
\pi_1^\ast &= \frac{\pi_1^\prime}{\sum \pi_1^\prime} = \frac{\pi_1^\prime}{1-\beta} \\
\pi_2^\ast &= \frac{\pi_2^\prime}{\sum \pi_2^\prime} = \frac{\pi_2^\prime}{1-\beta}
\end{align*}
We know that $\pi_1 = (\pi_1 \wedge \pi_2) + \pi_1^\prime$ and a similar statement holds for $\pi_2$, however we can now rewrite $\pi_1$ and $\pi_2$ as a sum of two densities:
\begin{align*}
\pi_1 &= \beta \pi_{1 \wedge 2} + (1-\beta) \pi_1^\ast \\
\pi_2 &= \beta \pi_{1 \wedge 2} + (1-\beta) \pi_2^\ast
\end{align*}
Note that 
\begin{align*}
d_{V}(\mathbf{\rho }_{1},\mathbf{\rho }_{2}) &= d_V(\pi_1 P,\pi_2 P) \\
&= d_V \Big( \left[\beta \pi_{1 \wedge 2} +(1-\beta)\pi_1^\ast\right] P,\left[\beta \pi_{1 \wedge 2} +(1-\beta)\pi_2^\ast\right] P \Big)\\
&\leq (1-\beta )d_{V}(\mathbf{\pi }_{1}^{\ast }P,\mathbf{\pi }_{2}^{\ast}P).
\end{align*}
%\begin{center}
%\includegraphics[scale=0.5]{handdrawn_pic}
%\end{center}
Now consider $\mathbf{\pi }_{1}^{\ast }P, \mathbf{\pi }_{2}^{\ast}P$ as different mixture mass functions on $Y$ given that $X$ takes each of its $n$ values, i.e. the different rows $\mathbf{p}_{i}$ of $P$. Let $\pi_1^\ast=(\pi_{11}^\ast,\pi_{12}^\ast,\ldots,\pi_{1n}^\ast)$, $\pi_2^\ast=(\pi_{21}^\ast,\pi_{22}^\ast,\ldots,\pi_{2n}^\ast)$ and the rows of P be $\mathbf{p}_i=(p_{i1},p_{i2},\ldots,p_{in})$. 
Note that by simple matrix multiplication on the entries we can now write 
\begin{align*}
\pi_1^\ast P &= \left( \left( \pi_{11}^\ast p_{11} + \pi_{12}^\ast p_{21}+\ldots+\pi_{1n}p_{n 1}\right) ,\left( \pi_{11}^\ast p_{12} + \pi_{12}^\ast p_{22}+\ldots+\pi_{1n}p_{n 2}\right),\ldots,\right. \\
&\left. \hspace{2cm}\left(\pi_{11}^\ast p_{1n} + \pi_{12}^\ast p_{2n}+\ldots+\pi_{1n}p_{n n}\right) \right) \\
&= \pi_{11}^\ast \left( p_{11},p_{12},\ldots,p_{1n}\right) + \pi_{12}^\ast \left(p_{21},p_{22},\ldots,p_{2 n}\right) + \ldots + \pi_{in}\left(p_{n 1},p_{n 2},\ldots,p_{nn}\right) \\
&= \pi_{11}^\ast \mathbf{p}_{1} + \pi_{12}^\ast \mathbf{p}_{2} + \ldots + \pi_{1n}^\ast \mathbf{p}_{n } \\
&= \sum_{i=1}^{n} \pi_{1i}^\ast \mathbf{p}_{i}
\end{align*}
Similarly,
\begin{equation*}
\pi_2^\ast P = \sum_{i=1}^{n} \pi_{2i}^\ast \mathbf{p}_i
\end{equation*}
which means we can now employ Lemma \ref{mixing_lemma}
\begin{align*}
d_{V}(\mathbf{\rho }_{1},\mathbf{\rho }_{2}) & \leq (1-\beta) d_V(\pi_1^\ast P,\pi_2^\ast P) \\
&\leq (1-\beta
)\sum_{i=1}^{n}\sum_{i^{\prime }=1}^{n}\pi _{1i}^{\ast }\pi _{2i^{\prime }}^{\ast }d_{V}(\mathbf{p}_{i^{\prime }},\mathbf{p}_{i}) \\
&=(1-\beta )\mathbf{\pi }_{1}D_{0}(P)\mathbf{\pi }_{2}^{T}
\end{align*}
So by definition of the diameter $d^+(P)$
\begin{align*}
d_{V}(\mathbf{\rho }_{1},\mathbf{\rho }_{2}) &\leq (1-\beta
)\sum_{i=1}^{n}\sum_{i^{\prime }=1}^{n}\pi _{1i}^{\ast }\pi _{2i^{\prime}}^{\ast }d^+(P) \\
&& \\
&=(1-\beta )d^+(P) \\
&=d^+(P)d_{V}(\mathbf{\pi }_{1},\mathbf{\pi }_{2})
\end{align*}
\end{proof}

\subsection{Proof of Lemma \ref{lemma14}}
\label{appendixB1}
Lemma \ref{lemma14} stated that for any two joint probability mass functions $p_{\boldsymbol{X,Y}}(\boldsymbol{x},\boldsymbol{y}), p_{\boldsymbol{X,Y}
}^{\prime }(\boldsymbol{x},\boldsymbol{y})$ over $\boldsymbol{X,Y}$ 
\begin{equation*}
d_{V}(p_{\boldsymbol{X,Y}}(\boldsymbol{x},\boldsymbol{y}),p_{\boldsymbol{X,Y}}^{\prime }(\boldsymbol{x},\boldsymbol{y}))\leq \inf \left\{ d_{V}(p_{\boldsymbol{X}}(\boldsymbol{x}),p_{\boldsymbol{X}}^{\prime }(\boldsymbol{x}))+\sup_{\boldsymbol{x}}d_{V}(p_{\boldsymbol{Y}|\boldsymbol{X}}(\boldsymbol{y
}|\boldsymbol{x}),p_{\boldsymbol{Y}|\boldsymbol{X}}^{\prime }(\boldsymbol{y}|\boldsymbol{x})),1\right\}.
\end{equation*}
\begin{proof}
Note
\begin{align*}
2d_{V}(p_{\boldsymbol{X,Y}}(\boldsymbol{x},\boldsymbol{y}),p_{\boldsymbol{X,Y}}^{\prime }(\boldsymbol{x},\boldsymbol{y})) &\triangleq \sum_{x,y}\left\vert p_{\boldsymbol{X,Y}}(\boldsymbol{x},\boldsymbol{y})-p_{\boldsymbol{X,Y}}^{\prime }(\boldsymbol{x},\boldsymbol{y})\right\vert  \\
&\triangleq \sum_{x,y}\left\vert p_{\boldsymbol{Y}|\boldsymbol{X}}(
\boldsymbol{y}|\boldsymbol{x})p_{\boldsymbol{X}}(\boldsymbol{x})-p_{
\boldsymbol{Y}|\boldsymbol{X}}^{\prime }(\boldsymbol{y}|\boldsymbol{x})p_{
\boldsymbol{X}}^{\prime }(\boldsymbol{x})\right\vert. 
\end{align*}
Let $r(\boldsymbol{y}|\boldsymbol{x})\triangleq p_{\boldsymbol{Y}|
\boldsymbol{X}}^{\prime }(\boldsymbol{y}|\boldsymbol{x})-p_{\boldsymbol{Y}|
\boldsymbol{X}}(\boldsymbol{y}|\boldsymbol{x})$. Then 
\begin{align*}
2d_{V}(p_{\boldsymbol{X,Y}}(\boldsymbol{x},\boldsymbol{y}),p_{\boldsymbol{X,Y}}^{\prime }(\boldsymbol{x},\boldsymbol{y})) &=\sum_{x,y}\left\vert p_{\boldsymbol{Y}|\boldsymbol{X}}(\boldsymbol{y}|
\boldsymbol{x})\left( p_{\boldsymbol{X}}(\boldsymbol{x})-p_{\boldsymbol{X}
}^{\prime }(\boldsymbol{x})\right) -r(\boldsymbol{y}|\boldsymbol{x})p_{
\boldsymbol{X}}^{\prime }(\boldsymbol{x})\right\vert  \\
&\leq \sum_{x}\left\{ \left\vert p_{\boldsymbol{X}}(\boldsymbol{x})-p_{
\boldsymbol{X}}^{\prime }(\boldsymbol{x})\right\vert \left( \sum_{y}p_{
\boldsymbol{Y}|\boldsymbol{X}}(\boldsymbol{y}|\boldsymbol{x})\right)
\right\} +\sum_{x,y}p_{\boldsymbol{X}}^{\prime }(\boldsymbol{x})\left\vert r(\boldsymbol{y}|\boldsymbol{x})\right\vert  \\
&=\sum_{x}\left\vert p_{\boldsymbol{X}}(\boldsymbol{x})-p_{\boldsymbol{X}
}^{\prime }(\boldsymbol{x})\right\vert +\sum_{x}\left\{ p_{\boldsymbol{X}}^{\prime }(\boldsymbol{x})\sum_{y}\left\vert r(\boldsymbol{y}|\boldsymbol{x})\right\vert \right\}  \\
&\leq \sum_{x}\left\vert p_{\boldsymbol{X}}(\boldsymbol{x})-p_{\boldsymbol{X}}^{\prime }(\boldsymbol{x})\right\vert +\sum_{y}\sup \left\vert r(
\boldsymbol{y}|\boldsymbol{x})\right\vert  \\
&\triangleq 2d_{V}(p_{\boldsymbol{X}}(\boldsymbol{x}),p_{\boldsymbol{X}
}^{\prime }(\boldsymbol{x}))+2\sup_{\boldsymbol{x}}d_{V}(p_{\boldsymbol{Y}|
\boldsymbol{X}}(\boldsymbol{y}|\boldsymbol{x}),p_{\boldsymbol{Y}|\boldsymbol{X}}^{\prime }(\boldsymbol{y}|\boldsymbol{x})).
\end{align*}
The result follows, since variation distance is, by definition, bounded by $1$. This simple result leads to the next useful result:
\end{proof}

\subsection{Proof of Lemma \ref{diameter_nonsingle_separator}}
\label{appendixB2}
Lemma \ref{diameter_nonsingle_separator} stated that 
$d^+(P_{\boldsymbol{Y|X}})\leq \inf \left\{ d^+(P_{\boldsymbol{Y}_{1}\boldsymbol{|X}})+d^+(P_{\boldsymbol{Y}_{2}\boldsymbol{|X,Y}_{1}}),1\right\}. $
\begin{proof}
In Lemma \ref{lemma14} substituting $p=p(y|x)$ and $p^\prime=p(y|x^\prime)$  gives us that for all values of $\boldsymbol{x,x}^{\prime }$
\begin{align*}
&d_{V}(p_{\boldsymbol{Y|X}}(\boldsymbol{y}_{1},\boldsymbol{y}_{2}|
\boldsymbol{x}),p_{\boldsymbol{Y|X}}(\boldsymbol{y}_{1},\boldsymbol{y}_{2}|\boldsymbol{x}^{\prime })) \\
&\leq d_{V}(p_{\boldsymbol{Y}_{1}|\boldsymbol{X}=x}(\boldsymbol{y}_{1}|
\boldsymbol{x}),p_{\boldsymbol{Y}_{1}|\boldsymbol{X}=x}(\boldsymbol{y}_{1}|\boldsymbol{x}^{\prime }))+\sup_{\left( \boldsymbol{y}_{1},\boldsymbol{x}\right) \boldsymbol{,}\left( \boldsymbol{y}_{1}^{\prime }\boldsymbol{,x}^{\prime }\right) }d_{V}(p_{\boldsymbol{Y}_{2}|\boldsymbol{Y}_{1},\boldsymbol{X}}(\boldsymbol{y}_{2}|\boldsymbol{y}_{1},\boldsymbol{x}),p_{\boldsymbol{Y}_{2}|\boldsymbol{Y}_{1},\boldsymbol{X}}(\boldsymbol{y}_{2}|\boldsymbol{y}_{1},\boldsymbol{x}^{\prime })).
\end{align*}
But by definition of the diameter
\begin{align*}
d^+(P_{\boldsymbol{Y}_{1}\boldsymbol{|X}}) &=\sup_{\boldsymbol{x},
\boldsymbol{x}^{\prime }}\left\{ d_{V}(p_{\boldsymbol{Y}_{1}|\boldsymbol{X}=x}(\boldsymbol{y}_{1}|\boldsymbol{x}),p_{\boldsymbol{Y}_{1}|\boldsymbol{X}=x}(\boldsymbol{y}_{1}|\boldsymbol{x}^{\prime })\right\},  \\
d^+(P_{\boldsymbol{Y}_{2}\boldsymbol{|X,Y}_{1}}) &=\sup_{\left( 
\boldsymbol{y}_{1},\boldsymbol{x}\right) ,\left( \boldsymbol{y}_{1}^{\prime},\boldsymbol{x}^{\prime }\right) }\left\{ d_{V}(p_{\boldsymbol{Y}_{2}|\boldsymbol{Y}_{1},\boldsymbol{X}}(\boldsymbol{y}_{2}|\boldsymbol{y}_{1},\boldsymbol{x}),p_{\boldsymbol{Y}_{2}|\boldsymbol{Y}_{1},\boldsymbol{X}}(\boldsymbol{y}_{2}|\boldsymbol{y}_{1},\boldsymbol{x}^{\prime })\right\}. 
\end{align*}
The result follows. 
\end{proof}
\end{appendices}

\end{document}